\documentclass[a4paper,11pt]{article}
\usepackage[margin=1in]{geometry}
\usepackage{amsfonts}
\usepackage{amsthm}
\usepackage{amsmath}
\usepackage{amssymb}
\usepackage{blkarray}
\usepackage[dvipsnames]{xcolor}
\usepackage{graphicx} 
\usepackage{subcaption}
\usepackage{mathtools}
\usepackage{underscore} 
\usepackage{nicematrix}
\NiceMatrixOptions{
code-for-first-row = \color{blue}\scriptstyle ,
code-for-last-row = \color{blue}\scriptstyle ,
code-for-first-col = \color{blue}\scriptstyle ,
code-for-last-col = \color{blue}\scriptstyle
}

\usepackage{newtxtext,newtxmath}

\usepackage{enumitem}
\usepackage{hyperref}
\usepackage{crossreftools}
\usepackage{tikz}
\usepackage{tikzit}
\usetikzlibrary{quantikz2} 
\usepackage{algorithm}
\usepackage{algpseudocodex}
\usepackage{bm}
\usepackage{stmaryrd}
\usepackage{anyfontsize}
\usepackage{dsfont} 
\usepackage{multicol}
\usepackage{xspace}
\usepackage{lipsum}
\usepackage{booktabs}
\usepackage{authblk}
\usepackage{pifont}
\usepackage{placeins} 

\usepackage[T1]{fontenc}

\tikzstyle{Large Black}=[fill=black, draw=black, shape=circle, tikzit category=GR, tikzit fill=black]
\tikzstyle{Large Empty}=[fill=white, draw=black, shape=circle, tikzit category=GR, tikzit fill=white]
\tikzstyle{Large Empty Box}=[fill=white, draw=black, shape=rectangle, tikzit category=GR]
\tikzstyle{Small Black}=[fill=black, draw=black, shape=circle, inner sep=0pt, minimum size=6pt, tikzit category=GR]
\tikzstyle{Small Empty}=[fill=white, draw=black, shape=circle, inner sep=0pt, minimum size=6pt, tikzit category=GR]
\tikzstyle{Small Empty Box}=[fill=white, draw=black, shape=rectangle, inner sep=0pt, minimum size=8pt, tikzit category=GR]
\tikzstyle{Small Grey}=[fill={rgb,255: red,119; green,119; blue,119}, draw={rgb,255: red,119; green,119; blue,119}, shape=diamond, tikzit category=GR, tikzit fill={rgb,255: red,119; green,119; blue,119}, inner sep=0pt, minimum size=6pt]
\tikzstyle{Small Grey Empty}=[draw={rgb,255: red,119; green,119; blue,119}, shape=circle, tikzit category=GR, tikzit fill={rgb,255: red,119; green,119; blue,119}, inner sep=0pt, minimum size=6pt]
\tikzstyle{GreyTEXT}=[align=center, text={rgb,255: red,119; green,119; blue,119}]
\tikzstyle{Tiny Black}=[fill=black, draw=black, shape=circle, inner sep=0pt, minimum size=4pt, tikzit category=GR]
\tikzstyle{BlackTEXT}=[align=center]

\tikzstyle{empty}=[-, fill={rgb,255: red,191; green,191; blue,191}, draw={rgb,255: red,191; green,191; blue,191}]
\tikzstyle{brace edge}=[-, tikzit draw=blue, decorate, decoration={brace,amplitude=1mm,raise=-1mm}]
\tikzstyle{grey}=[-, draw={rgb,255: red,119; green,119; blue,119}, densely dashed]
\tikzstyle{directed}=[->]
\tikzstyle{thick}=[-, line width=1mm]

\graphicspath{ {./img/} }

\makeatletter
\newcommand{\optionaldesc}[2]{%
  \phantomsection
  #1\protected@edef\@currentlabel{#1}\label{#2}%
}
\makeatother

\theoremstyle{definition}
\newtheorem{theorem}{Theorem}[section]
\newtheorem{corollary}[theorem]{Corollary}
\newtheorem{lemma}[theorem]{Lemma}
\newtheorem*{lemma*}{Lemma}
\newtheorem*{proposition*}{Proposition}
\newtheorem*{remark*}{Remark}

\newtheorem{definition}[theorem]{Definition}

\newtheorem{observation}[theorem]{Observation}

\newtheorem*{rep@theorem}{\rep@title}
\newcommand{\newreptheorem}[2]{%
	\newenvironment{rep#1}[1]{%
    \def\rep@title{#2 \ref{##1} (restated)}%
		\begin{rep@theorem}}%
		{\end{rep@theorem}}}
\newreptheorem{thm}{Theorem}

  {\list{}{\leftmargin=#1\rightmargin=#1}\item[]}%
  {\endlist}

\newcommand{\RM}{\mathrm}
\newcommand{\counting}[1][]{\mathord{\#}\mathrm #1} 

\newcommand{\customleft}{\left[\!\!\left\langle} 
\newcommand{\customright}{\right\rangle\!\!\right]} 
\newcommand{\bigO}{\mathcal{O}} 
\newcommand{\rowbasechange}[2]{{#1}_{#2}}
\newcommand{\colbasechange}[2]{{#1}^{#2}}
\newcommand{\starpart}{P}
\newcommand{\solved}{S}
\newcommand{\tosolve}{L}
\newcommand{\LS}{K_{LS}}
\newcommand{\ILS}{K_{ILS}}
\newcommand{\proc}[1]{#1}
\newcommand{\expl}[1]{{\emph{Explanation of the step:} #1}}
\DeclareMathOperator{\oddop}{Odd}
\newcommand{\odd}[1]{\oddop\mathopen{}\left(#1\right)\mathclose{}} 
\newcommand{\codd}[1]{\oddop\mathopen{}\customleft #1\customright\mathclose{}} 
\DeclareMathOperator{\rank}{rank}

\newcommand{\ind}[2]{\mathds{1}_{#1}^{#2}} 
\DeclareMathOperator{\support}{supp} 
\newcommand{\supp}[1]{\support\left(#1\right)} 
\newcommand{\symd}{\mathbin{\Delta}\xspace} 
\newcommand{\trl}{\triangleleft} 
\newcommand{\meas}[1]{\Lambda_{#1}} 
\newcommand{\measplanar}{\meas{planar}} 
\newcommand{\measPauli}{\meas{Pauli}} 
\newcommand{\comp}[1]{\bar{#1}} 
\newcommand{\ld}{\lambda}
\newcommand{\sse}{\subseteq}
\newcommand{\abs}[1]{\left| #1 \right|} 
\newcommand{\FF}{\mathbb{F}}

\newcommand{\someset}{\mathcal{A}}
\newcommand{\otherset}{\mathcal{D}}

\newcommand{\A}{{A}}
\newcommand{\Azz}{\A_{00}}
\newcommand{\Azo}{\A_{01}}
\newcommand{\Aoz}{\A_{10}}
\newcommand{\Aoo}{\A_{11}}
\newcommand{\Czz}{C_{00}}
\newcommand{\I}{Id}
\let\dl=\delta
\newcommand{\Xlike}{\mathcal{X}}
\newcommand{\Zlike}{\mathcal{Z}}
\newcommand{\internal}{\mathcal{B}}
\newcommand{\paulis}{\mathcal{P}}
\newcommand{\planar}{\mathcal{L}}
\newcommand{\fullN}{S}
\newcommand{\LOG}{labelled open graph}

\newcommand{\triplecase}[6]{
    #1 \quad & \text{and} \quad #2, \text{or}\\
    #3 \quad & \text{and} \quad #4, \text{or}\\
    #5 \quad & \text{and} \quad #6.
}

\title{An algebraic interpretation of Pauli flow, leading to faster flow-finding algorithms}

\author[1]{Piotr Mitosek}
\author[2]{Miriam Backens}
\affil[1]{School of Computer Science, University of Birmingham, UK}
\affil[2]{Inria \& Loria, Nancy, France}
\affil[1]{\textit {pbm148@student.bham.ac.uk}}

\begin{document}

\maketitle

\begin{abstract}
    \noindent The one-way model of quantum computation is an alternative to the circuit model. A one-way computation is driven entirely by successive adaptive measurements of a pre-prepared entangled resource state. For each measurement, only one outcome is desired; hence a fundamental question is whether some intended measurement scheme can be performed in a robustly deterministic way. So-called flow structures witness robust determinism by providing instructions for correcting undesired outcomes. Pauli flow is one of the broadest of these structures and has been studied extensively. It is known how to find flow structures in polynomial time when they exist; nevertheless, their lengthy and complex definitions often hinder working with them.

    We simplify these definitions by providing a new algebraic interpretation of Pauli flow. This involves defining two matrices arising from the adjacency matrix of the underlying graph: the flow-demand matrix $M$ and the order-demand matrix $N$. We show that Pauli flow exists if and only if there is a right inverse $C$ of $M$ such that the product $NC$ forms the adjacency matrix of a directed acyclic graph. From the newly defined algebraic interpretation, we obtain $\bigO(n^3)$ algorithms for finding Pauli flow, improving on the previous $\bigO(n^4)$ bound for finding generalised flow, a weaker variant of flow, and $\bigO(n^5)$ bound for finding Pauli flow. We also introduce a first lower bound for the Pauli flow-finding problem, by linking it to the matrix invertibility and multiplication problems over $\FF_2$.
\end{abstract}

\section{Introduction}\label{sec:introduction}

Quantum computing offers not only computational capabilities that go beyond those of classical computers, it also gives rise to entirely non-classical models of computation such as the measurement-based one-way model \cite{raussendorfOneWayQuantumComputer2001, raussendorfOnewayQuantumComputer2002}.
The one-way model makes use of quantum entanglement and of the property that quantum measurements change the state which is being measured.

A one-way computation thus consists of two parts:
Firstly, an entangled graph state is prepared as a resource.
There are universal families of graph states that enable any computation to be performed as long as the state is large enough \cite{raussendorfMeasurementbasedQuantumComputation2003,broadbentUniversalBlindQuantum2009,mhallaGraphStates2013}.
Secondly, the computation itself proceeds via successive adaptive single-qubit measurements on this resource state\footnote{It is also possible to intertwine the two steps, e.g.\ to re-use the same physical systems as different qubits over the course of a single computation, but for simplicity we will consider state preparation and measurement phase separately here.}.
Individual measurements are non-deterministic and only one of the two possible outcomes of a single-qubit measurement drives the computation in the desired direction.
Nevertheless, with an appropriate combination of resource\ and measurements, it is possible to modify later measurements to compensate for having seen an undesired outcome earlier.
In that way the computation as a whole is deterministic in a suitable sense, referred to as `strong uniform stepwise determinism' \cite{danosDeterminismOnewayModel2006,browneGeneralizedFlowDeterminism2007} or `robust determinism' \cite{perdrixDeterminismComputationalPower2017,mhallaCharacterisingDeterminismMBQCs2022}.

This combination of resource preparation and measurement means that the one-way model is particularly well suited to a client-server split: for example, the blind quantum computing scheme allows computations to be securely delegated to one or more quantum servers while requiring no or very limited quantum computational abilities on the part of the client \cite{broadbentUniversalBlindQuantum2009,fitzsimonsUnconditionally2017,kashefiMultiparty2017}.
Moreover, the one-way model also links well with quantum error-correcting codes such as the surface code \cite{raussendorfFaultTolerant2006,kashefiInformation2009}.

In addition to these practical applications, the one-way model is also of theoretical significance: the more flexible structure of graph states as compared to quantum circuits means it is often easier to optimise computations expressed in the one-way model \cite{broadbentParallelizingQuantumCircuits2009,duncanGraphtheoreticSimplificationQuantum2020,backensThereBackAgain2021,staudacherReducing2QuBitGate2023}.
One-way computations are formally described by \emph{measurement patterns} using the measurement calculus : sequences of commands for the preparation of the resource state, the measurements, and the potentially-required corrections \cite{danosMeasurementCalculus2007,danosExtendedMeasurementCalculus2009}.
Yet all the information required to determine the linear map implemented by a one-way computation, as well as to check whether it can be implemented deterministically, can also be expressed as a so-called labelled open graph together with an assignment of measurement angles for each non-output qubit.
The \emph{labelled open graph} $\Gamma$ underlying a one-way computation consists of:
\begin{itemize}[noitemsep,topsep=0pt,parsep=0pt,partopsep=0pt]
 \item the mathematical graph $G=(V,E)$ corresponding to the resource graph state,
 \item two subsets $I,O\sse V$ corresponding to the input and output qubits of the computation, and
 \item a function $\ld:\comp{O}\to\{X,Y,Z,XY,XZ,YZ\}$ assigning to each non-output qubit one of the given six measurement labels. (For an example, see the left-most part of Figure~\ref{fig:loG}.)
\end{itemize}
A one-way computation is robustly deterministic computation if and only if the associated labelled open graph satisfies a property called \emph{Pauli flow} \cite{browneGeneralizedFlowDeterminism2007,mhallaCharacterisingDeterminismMBQCs2022} (the measurement angles do not affect the presence or absence of flow).
If all the measurement labels are in the set $\{XY,XZ,YZ\}$ -- called `planar measurements' -- the necessary and sufficient condition is called \emph{(extended) gflow}\footnote{There is also a type of flow called `extended Pauli flow', but the relationship between `extended gflow' and `gflow' is different to that between `extended Pauli flow' and `Pauli flow': the plain term \emph{gflow} denotes either an extended gflow with only $XY$-measurements or it is a synonym for extended gflow, whereas extended Pauli flow is a proper generalisation of Pauli flow.
For measurement patterns (which implicitly define an order on the measurement operations) instead of labelled open graphs, it is only extended Pauli flow that is necessary \cite{mhallaCharacterisingDeterminismMBQCs2022}.} \cite{browneGeneralizedFlowDeterminism2007}.
Each measurement pattern has a unique underlying labelled open graph, but a labelled open graph generally supports multiple measurement patterns with different time orderings or correction procedures.
Up to information about the measurement angles (which is irrelevant to the question of robust determinism), the labelled open graph already determines which linear map is implemented.
For this reason, it is often useful to work directly with labelled open graphs instead of with measurement patterns.

Thus, a labelled open graph is said to have extended Pauli flow (or any other type of flow) if at least one of the measurement patterns it supports satisfies the extended Pauli flow property (or the relevant flow property).
It turns out that, whenever a labelled open graph has extended Pauli flow, it also supports measurement patterns with the more restricted properties of \emph{Pauli flow} and \emph{focused Pauli flow} \cite{simmonsRelatingMeasurementPatterns2021,mhallaCharacterisingDeterminismMBQCs2022}.
Similarly, in the case of planar measurements, whenever a labelled open graph has extended gflow it also has \emph{focused gflow} \cite{mhallaWhichGraphStates2014a,backensThereBackAgain2021}.
`Focused' here means that corrections are restricted in certain ways that will be made precise in Section~\ref{sec:background}.

There exist algorithms for finding different types of flow on a labelled open graph, or determining that no flow of the given type exists.
If all measurements are of $XY$-type, a so-called \emph{causal flow} can be found in polynomial time \cite{debeaudrapCompleteAlgorithmFind2007,deBeaudrapFinding2008} (yet this causal flow is not necessary for robust determinism \cite{browneGeneralizedFlowDeterminism2007}).
The algorithm for finding causal flow was improved to $\bigO(m)$ in \cite{mhallaFindingOptimalFlows2008a}, where $m$ is the number of edges.
If all measurements are planar, the gflow-finding algorithm runs in $\bigO(n^4)$, where $n$ is the number of vertices, and it always returns a focused gflow \cite{mhallaFindingOptimalFlows2008a,backensThereBackAgain2021}.
In the more general case, the Pauli flow-finding algorithm uses $\bigO(n^5)$ operations and is not guaranteed to return a focused flow \cite{simmonsRelatingMeasurementPatterns2021}.

For labelled open graphs with only $XY$-measurements, a particularly nice characterisation of focused gflow is known: if it exists, the focused gflow  corresponds to a right inverse of a certain submatrix of the graph adjacency matrix \cite{mhallaWhichGraphStates2014a}.
This inverse encodes the correction procedure and a partial time order: for the latter reason, it must be interpretable as a directed acyclic graph (DAG).
If the numbers of inputs and outputs are equal, then the above submatrix is square, hence any right inverse is a two-sided inverse and focused gflow on such labelled open graphs is unique and reversible \cite{mhallaWhichGraphStates2014a}.
The uniqueness of focused gflow for $\abs{I}=\abs{O}$ extends to the case where all planar measurements are allowed \cite{backensThereBackAgain2021}.
Similarly, focused Pauli flow is known to be unique whenever $\abs{I}=\abs{O}$ \cite{simmonsRelatingMeasurementPatterns2021}, although both results were shown using other techniques.
An extension of the linear algebraic technique to labelled open graphs with only $X$ and $Z$ measurements\footnote{The same paper also mentioned a to-be-published version for $XY$ and $X$ measurements: this paper is that publication.} was used by one of the authors to find measurement labellings compatible with Pauli flow for unlabelled open graphs \cite{mitosekPauliFlowOpen2024}.

In this paper, we give a linear algebraic characterisation for focused Pauli flow and for all labelled open graphs with $\abs{I}\leq\abs{O}$.
We show how to construct a `flow-demand matrix', which takes the place of the submatrix of the adjacency matrix.
A right inverse $C$ of this matrix specifies a correction procedure.
To see whether this correction procedure is compatible with a Pauli flow, we construct an `order-demand matrix' whose matrix product with $C$ is required to be interpretable as a DAG.

Given equal numbers of inputs and outputs, the flow-demand matrix is square, so $C$ (if it exists) is unique.
Therefore, we obtain an alternative proof that a focused Pauli flow is also unique when it exists.
Moreover, we prove that under these conditions focused Pauli flow is reversible, in the sense that the labelled open graph where the sets of inputs and outputs are swapped has focused Pauli flow if and only if the original one does.
For $XY$-measurements only, this reversibility holds in a very strong sense \cite{mhallaWhichGraphStates2014a}: if the original gflow has correction matrix $C$, then the reversed gflow has correction matrix $C^T$.
We show this strong reversibility extends to Pauli flow on labelled open graphs with $X$ and $Y$ measurements too.
Yet if there are any $XZ$, $YZ$, or $Z$ measurements, the relationship between the correction functions of the original and reversed Pauli flow becomes more complicated.
Nevertheless, the partial order relationships among planar-measured internal vertices are exactly inverted in the reversed Pauli flow.

Furthermore, we consider the case where there are more outputs than inputs, so the flow-demand matrix is not square and may have multiple right inverses.
A change of basis nevertheless allows us to efficiently find a $C$ that satisfies the DAG property with the order-demand matrix, or to determine that no such $C$ exists.

Overall, our linear-algebraic Pauli flow-finding algorithm runs in time $\bigO(n^3)$ on $n$-vertex labelled open graphs, both when $\abs{I}=\abs{O}$ and when $\abs{I}<\abs{O}$.
For the case $\abs{I}=\abs{O}$, we reduce from the problem of finding a matrix inverse over $\FF_2$ to show that the complexity must be at least $\Omega(n^2)$.
The worse complexity of finding Pauli flow (compared to gflow) had been an obstacle to using Pauli flow in circuit optimisation~\cite{staudacherReducing2QuBitGate2023}: this obstacle is now removed.
The faster flow-finding algorithm may also have practical impact on algorithms for translating one-way computations into quantum circuits, for a more in-depth discussion of this see Section~\ref{sec:summary and future work}.

In a labelled open graph with $\abs{I}<\abs{O}$, there exist sets of correction instructions whose net effect is trivial.
Such correction instructions correspond to sets of vertices called \emph{focused sets}, which were previously considered by Simmons \cite{simmonsRelatingMeasurementPatterns2021}.
We use the linear algebra perspective to gain more insight into the properties of focused sets and under which conditions they can be used to transform one focused Pauli flow into another focused Pauli flow.

The remainder of the paper is structured as follows. In Section~\ref{sec:background}, we present the background material, formally introduce Pauli flow and give some examples. Next, in Section~\ref{sec:new algebraic interpretation}, we introduce the new algebraic interpretation and show its properties. Then, in Section~\ref{sec:algorithms}, we use the algebraic interpretation to develop improved flow-finding algorithms. We conclude the paper in Section~\ref{sec:summary and future work}, where we also discuss ideas for future work.

\section{Background}\label{sec:background}

In this section, we introduce the essential definitions, including labelled open graphs, Pauli flow, and the focusing conditions.
We also give relevant examples.
General knowledge of the one-way model of measurement-based quantum computation (MBQC) is assumed, to refresh their knowledge the reader may consult e.g.~\cite{browneGeneralizedFlowDeterminism2007}.
As mentioned in the introduction, we do not use the measurement calculus but instead work directly with labelled open graphs.

To begin, an open graph is a graph with distinguished sets of input and output vertices.

\begin{definition}[Open graph]
    An \textit{open graph} is a triple $(G,I,O)$, where:\begin{itemize}
        \item $G = (V,E)$ is an undirected simple graph consisting of vertices $V$ and edges $E$, 
        \item $I \subseteq V$ is a set of \textit{inputs}, and
        \item $O \subseteq V$ is a set of \textit{outputs}.
    \end{itemize}
    When $(G,I,O)$ is fixed, we use the following notation:\begin{itemize}
        \item $\comp{O} := V \setminus O$ is the set of \textit{non-outputs},
        \item $\comp{I} := V \setminus I$ is the set of \textit{non-inputs},
        \item $n := |V|$ is the number of vertices,
        \item $n_I := |I|$, is the numbers of inputs, and
        \item $n_O := |O|$, is the number of outputs.
    \end{itemize}
    The \textit{odd neighbourhood} of a set of vertices $\someset$ is the set of vertices $v\in V$ which are adjacent to an odd number of elements of $\someset$, i.e.:
    \begin{equation*}
        \odd{\someset} := \{ v \in V : \abs{\{ u \in \someset \mid vu \in E \}} \text{ is odd} \}.
    \end{equation*}

    \noindent A \textit{closed odd neighbourhood} of a set of vertices $\someset$ is defined as follows, where $\symd$ stands for a symmetric difference of sets, i.e.\ $\otherset\symd \otherset' = (\otherset\cup \otherset') \setminus (\otherset\cap \otherset')$:
    \begin{align*}
        \codd{\someset} &:= \odd{\someset} \symd \someset.
    \end{align*}
\end{definition}

Having defined open graphs, which relate to the resource graph states of a measurement-based computation, we now move on to considering the measurements.
The one-way model allows measurements in three planes of the block sphere, spanned by the eigenstates of pairs of Pauli operators -- these are $XY$, $YZ$, and $XZ$ planar measurements. The intersections between the planes, the Pauli measurements $X$, $Y$, and $Z$, are considered separately for a total of six measurement types.

\begin{definition}[Measurement labelling]\label{def:labelling}
    A \textit{measurement labelling} for an open graph $(G,I,O)$ is a function:\begin{equation*}
        \lambda \colon \comp{O} \to \{ X, Y, Z, XY, XZ, YZ \}
    \end{equation*}
    \noindent satisfying $\forall i \in I \setminus O . \lambda(i) \in \{ X, Y, XY \}$.
    
    A \textit{labelled open graph} is a quadruple $(G,I,O,\lambda)$ where $(G,I,O)$ is an open graph and $\lambda$ is a measurement labelling on this open graph.

    We refer to the set of planar measured vertices as $\measplanar$, and to the set of Pauli measured vertices as $\measPauli$:\begin{align*}
        \measplanar &= \left\{ v \in \comp{O} \mid \lambda(v) \in \{ XY, YZ, XZ \} \right\}\\
        \measPauli  &= \left\{ v \in \comp{O} \mid \lambda(v) \in \{ X, Y, Z \} \right\}.
    \end{align*}
\end{definition}

See Figure~\ref{fig:loG} for an example of a labelled open graph; this will be a running example to illustrate the different concepts defined throughout this paper.

We also need to specify measurement angles to turn a labelled open graph into an implementable quantum computation.
However, we will soon define Pauli flow which enables and witnesses robustly deterministic computation.
Here, `robustly' means, among other things, that the computation must be deterministic for any choice of the angles.
Hence, we ignore the measurement angles throughout this paper.

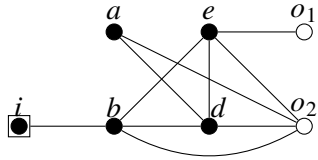
\begin{figure}
    \centering
    \begin{subfigure}[b]{0.5\textwidth}
        \begin{equation*}
            \begin{tikzpicture}
	\begin{pgfonlayer}{nodelayer}
		\node [style=Small Empty Box] (0) at (-3.75, -1.25) {};
		\node [style=Small Black] (1) at (-3.75, -1.25) {};
		\node [style=Small Black] (2) at (-1.25, -1.25) {};
		\node [style=Small Black] (3) at (-1.25, 1.25) {};
		\node [style=Small Black] (4) at (1.25, 1.25) {};
		\node [style=Small Black] (5) at (1.25, -1.25) {};
		\node [style=Small Empty] (8) at (3.75, 1.25) {};
		\node [style=Small Empty] (9) at (3.75, -1.25) {};
		\node [style=none] (10) at (-1.25, 1.75) {$a$};
		\node [style=none] (11) at (1.25, 1.75) {$e$};
		\node [style=none] (13) at (3.75, 1.75) {$o_1$};
		\node [style=none] (14) at (3.75, -0.75) {$o_2$};
		\node [style=none] (16) at (1.5, -0.75) {$d$};
		\node [style=none] (17) at (-1.25, -0.75) {$b$};
		\node [style=none] (18) at (-3.75, -0.75) {$i$};
		\node [style=none] (19) at (4.75, 0) {};
		\node [style=none] (20) at (-4.75, 0) {};
		\node [style=none] (21) at (0, 2.75) {};
		\node [style=none] (22) at (0, -2.75) {};
	\end{pgfonlayer}
	\begin{pgfonlayer}{edgelayer}
		\draw (0) to (2);
		\draw (2) to (4);
		\draw (5) to (4);
		\draw (3) to (5);
		\draw (3) to (9);
		\draw [bend right] (2) to (9);
		\draw (4) to (8);
		\draw (2) to (5);
		\draw (5) to (9);
		\draw (4) to (9);
	\end{pgfonlayer}
\end{tikzpicture}
        \end{equation*}
    \end{subfigure}
    \begin{subfigure}[b]{0.4\textwidth}
        $\begin{array}{|c|c|c|}
        \toprule
        v & \lambda(v) & \text{Neighbourhood of }v \\
        \midrule
        i & XY & b \\
        \midrule
        a & XZ & d, o_2 \\
        b & Y & i, e, d, o_2 \\
        e & XY & b, d, o_1, o_2 \\
        d & Z &  a, b, c, o_2 \\
        \midrule
        o_1 & & e \\
        o_2 & & a, b, e, d \\
        \bottomrule
        \end{array}$
    \end{subfigure}
    \caption{Example of a labelled open graph with one input $i$ and two outputs $o_1, o_2$. (We skip the letter~$c$ when naming vertices to avoid confusion with the soon-to-be-defined correction function. The graph is a modification of \cite[Example D.18]{simmonsRelatingMeasurementPatterns2021}.)}
    \label{fig:loG}
\end{figure}

Due to the nature of quantum measurements, measuring a qubit may give two outcomes of which only one is desired. In the case of obtaining the undesired outcome, one must correct it.
In the one-way model, these corrections consist of applying Pauli $X$ and $Z$ operations so that after the application the system is in the state it would have been in if the desired outcome was obtained. This is achieved by choosing the Pauli operations in such a way, that together with the undesired outcome (which itself corresponds to the application of Pauli operation to the measured qubit), all the Paulis form a graph stabilizer \cite{danosDeterminismOnewayModel2006,browneGeneralizedFlowDeterminism2007,backensThereBackAgain2021}. When such corrections are possible, the computation given by the labelled open graph is called deterministic. It is important to ensure that the desired computation is deterministic, as otherwise, it may be impossible to perform in practice. There are different useful variants of determinism (see \cite{mhallaWhichGraphStates2014a}). For instance, uniform determinism means that the correction procedure is independent of the choice of measurement angles. The most widely used of these variants, `strong uniform stepwise determinism' on a labelled open graph is captured by the existence of Pauli flow \cite{mhallaCharacterisingDeterminismMBQCs2022}.
This flow consists of a correction function and a partial order. The correction function for each vertex defines the set of vertices to which one must apply $X$ in case of the undesired outcome, with the set of vertices requiring a $Z$ given by its odd neighbourhood.
The partial order is the order according to which vertices must be measured for the corrections to be possible.

\begin{definition}[Correction function]
    Let $(G,I,O,\lambda)$ be a labelled open graph. A \textit{correction function} for $(G,I,O,\lambda)$ is a function $c \colon \comp{O} \to \mathcal{P}(\comp{I})$.
\end{definition}

\begin{definition}[Pauli flow {\cite[Definition 5]{browneGeneralizedFlowDeterminism2007}}]\label{def:Pauli-flow} 
    A \textit{Pauli flow} for a labelled open graph $(G,I,O,\lambda)$ is a pair $(c,\prec)$, where $c$ is a correction function for $(G,I,O,\lambda)$ and $\prec$ is a strict partial order on $\comp{O}$, such that for all $u \in \comp{O}$:\begin{enumerate}
        \item[{\crtcrossreflabel{(P1)}[P1]}] $\forall v \in c(u) . u \ne v \wedge \lambda(v) \notin \{ X, Y \} \Rightarrow u \prec v$
        \item[{\crtcrossreflabel{(P2)}[P2]}] $\forall v \in \odd{c(u)} . u \ne v \wedge \lambda(v) \notin \{ Y, Z \} \Rightarrow u \prec v$
        \item[{\crtcrossreflabel{(P3)}[P3]}] $\forall v \in \comp{O} . \neg (u \prec v) \wedge u \ne v \wedge \lambda(v) = Y \Rightarrow v\notin\codd{c(u)}$
        \item[{\crtcrossreflabel{(P4)}[P4]}] $\lambda(u) = XY \Rightarrow u \notin c(u) \wedge u \in \odd{c(u)}$
        \item[{\crtcrossreflabel{(P5)}[P5]}] $\lambda(u) = XZ \Rightarrow u \in c(u) \wedge u \in \odd{c(u)}$
        \item[{\crtcrossreflabel{(P6)}[P6]}] $\lambda(u) = YZ \Rightarrow u \in c(u) \wedge u \notin \odd{c(u)}$
        \item[{\crtcrossreflabel{(P7)}[P7]}] $\lambda(u) = X \Rightarrow u \in \odd{c(u)}$
        \item[{\crtcrossreflabel{(P8)}[P8]}] $\lambda(u) = Z \Rightarrow u \in c(u)$
        \item[{\crtcrossreflabel{(P9)}[P9]}] $\lambda(u) = Y \Rightarrow u \in \codd{c(u)}$.
    \end{enumerate}
    The sets $c(v)$ for $v \in \comp{O}$ are called the \textit{correction sets}.
\end{definition}

See Figure~\ref{fig:pf} for an example of Pauli flow.

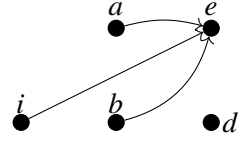
\begin{figure}
    \centering
    \begin{subfigure}[b]{0.4\textwidth}
        $\begin{array}{|c|c|c|c|}
            \toprule
             & \lambda & c(v) & \odd{c(v)} \\
            \midrule
            i & XY & b, e, o1 & i, b, o1 \\
            a & XZ & a, e, o1, o2 & {a, d, o1} \\
            b & Y & e & b, d, o1, o2 \\
            e & XY & o1 & e \\
            d & Z & d, o2 & d, o2 \\
            \bottomrule
        \end{array}$
    \end{subfigure}
    \begin{subfigure}[b]{0.5\textwidth}
        \begin{equation*}
            \begin{tikzpicture}
	\begin{pgfonlayer}{nodelayer}
		\node [style=Small Black] (0) at (-2.5, -1.25) {};
		\node [style=Small Black] (1) at (0, -1.25) {};
		\node [style=Small Black] (2) at (0, 1.25) {};
		\node [style=Small Black] (3) at (2.5, 1.25) {};
		\node [style=Small Black] (4) at (2.5, -1.25) {};
		\node [style=none] (7) at (0, 1.75) {$a$};
		\node [style=none] (8) at (2.5, 1.75) {$e$};
		\node [style=none] (11) at (3, -1.25) {$d$};
		\node [style=none] (12) at (0, -0.75) {$b$};
		\node [style=none] (13) at (-2.5, -0.75) {$i$};
		\node [style=none] (14) at (3.5, 0) {};
		\node [style=none] (15) at (-3.5, 0) {};
		\node [style=none] (16) at (0, 2.75) {};
		\node [style=none] (17) at (0, -2.75) {};
	\end{pgfonlayer}
	\begin{pgfonlayer}{edgelayer}
		\draw [style=directed, bend right] (1) to (3);
		\draw [style=directed, bend left=15] (2) to (3);
		\draw [style=directed] (0) to (3);
	\end{pgfonlayer}
\end{tikzpicture}
        \end{equation*}
    \end{subfigure}
    \caption{A correction function and a partial order forming Pauli flow for the labelled open graph from Figure~\ref{fig:loG}. The partial order is given as a directed acyclic graph.}
    \label{fig:pf}
\end{figure}

There may be many different Pauli flows on one labelled open graph. However, the existence of any Pauli flow implies the existence of a special type of Pauli flow called focused Pauli flow.
A focused flow must satisfy the following focusing conditions.

\begin{definition}[Focused Pauli flow {\cite[part of Definition 4.3]{simmonsRelatingMeasurementPatterns2021}}]\label{def:focused} 
    The Pauli flow $(c,\prec)$ is \textit{focused} when for all $v \in \comp{O}$ the following hold:\begin{enumerate}
        \item[{\crtcrossreflabel{(F1)}[F1]}] $\forall w \in (\comp{O} \setminus \{ v \}) \cap c(v) . \lambda(w) \in \{ XY, X, Y \}$
        \item[{\crtcrossreflabel{(F2)}[F2]}] $\forall w \in (\comp{O} \setminus \{ v \}) \cap \odd{c(v)} . \lambda(w) \in \{ XZ, YZ, Y, Z \}$
        \item[{\crtcrossreflabel{(F3)}[F3]}] $\forall w \in (\comp{O} \setminus \{ v \}) . \lambda(w) = Y \Rightarrow w \notin \codd{c(v)}$, note $w \notin \codd{c(v)}$ is equivalent to $w \in c(v) \Leftrightarrow w \in \odd{c(v)}$.
    \end{enumerate}
\end{definition}

\begin{theorem}[{\cite[Lemma 4.6]{simmonsRelatingMeasurementPatterns2021}}]\label{foc flow exists}
    For any open labelled graph, if a Pauli flow exists then there also exists a focused Pauli flow.
\end{theorem}

The focusing conditions are particularly useful in the case of an equal number of inputs and outputs (i.e.\ when $n_I = n_O$), as then the focused Pauli flow is unique if it exists. When the number of inputs and outputs differ, the focused flow is not necessarily unique.
The Pauli flow in the example of Figure~\ref{fig:pf} is focused.
If we change $c(i)$ to $\{ b, e \}$, then $\odd{c(i)} = \{ i, b, e, o_1 \}$. Keeping all other correction sets and partial order the same, we obtain a new Pauli flow. However, this flow is not focused -- it fails \ref{F2} as $e \in \odd{c(i)}$ and $\ld(e) = XY$.

The usual Pauli flow definition stated above describes the correction function and the partial order simultaneously. However, it will be more convenient to consider the two aspects separately, as captured in the following definitions.

\begin{definition}[Focused correction function]
    Let $(G,I,O,\lambda)$ be a labelled open graph and $c$ a correction function. We say that $c$ is \textit{focused} when it satisfies \ref{F1}, \ref{F2}, and \ref{F3}.
\end{definition}

\begin{definition}[Minimal order]
    Let $\Gamma = (G,I,O,\lambda)$ be a labelled open graph and let $(c,\prec)$ be a Pauli flow on $\Gamma$. We call $\prec$ \textit{minimal} when no proper subset of $\prec$ forms a Pauli flow with $c$.
\end{definition}

\begin{definition}[Extensive correction function]
    Let $\Gamma = (G,I,O,\lambda)$ be a labelled open graph and let $c$ be a correction function on $\Gamma$. We call $c$ \textit{extensive} when there exists $\prec$ such that $(c,\prec)$ is a Pauli flow on $\Gamma$. We call $c$ \textit{focused extensive} when there exists $\prec$ such that $(c,\prec)$ is a focused Pauli flow on $\Gamma$.
\end{definition}

We are not aware of the above definition for Pauli flow appearing in the literature. However, an analogous version for the case of only $XY$ planar measurements was considered in \cite{mhallaWhichGraphStates2014a}.

In the problem of finding Pauli flow, instead of looking for correction function and partial order, it now suffices to look only for a (focused) extensive correction function. This may not look like a simplification at first. However, given a correction function $c$ we can define a minimal relation $\trl_c$ induced by the Pauli flow definition, and consider whether this relation extends to a partial order.

\begin{definition}[Induced relation]
    Let $\Gamma = (G,I,O,\lambda)$ be a labelled open graph and let $c$ be a correction function on $\Gamma$. The \textit{induced relation} $\trl_c$ is the minimal relation on $\comp{O}$ implied by \ref{P1}, \ref{P2}, \ref{P3}. That is, $u \trl_c v$ if and only if at least one of the following holds:\begin{itemize}
        \item $v \in c(u) \wedge u \ne v \wedge \lambda(v) \notin \{ X, Y \}$ (corresponding to \ref{P1}),
        \item $v \in \odd{c(u)} \wedge u \ne v \wedge \lambda(v) \notin \{ Y, Z \}$ (corresponding to \ref{P2}),
        \item $v \in \codd{c(u)} \wedge u \ne v \wedge \lambda(v) = Y$ (corresponding to \ref{P3}).
    \end{itemize}
    We denote the transitive closure of $\trl_c$ as $\prec_c$.
\end{definition}

For example, the partial order given in Figure~\ref{fig:pf} is the relation $\trl_c$ for the correction function $c$ from the same figure. It happens to also be $\prec_c$, as it is already transitive.

In the following we prove some useful properties of $\trl_c$: Firstly, the transitive closure $\prec_c$ of the relation $\trl_c$ induced by the correction function $c$ is minimal. Further, if $c$ forms Pauli flow with any partial order, it also does so with $\trl_c$.

\begin{lemma}[Minimal order containment]\label{minimal order containment}
    Let $\Gamma = (G,I,O,\lambda)$ be a labelled open graph and let $(c,\prec)$ be a Pauli flow. Then ${\prec_c} \subseteq {\prec}$.
\end{lemma}

\begin{proof}
    By construction, $\trl_c$ is the minimal relation implied by \ref{P1}, \ref{P2}, \ref{P3}. Thus, any relation satisfying these conditions must contain $\trl_c$, thus $\trl_c \subseteq \prec$. Further, $\prec$ is a strict partial order, and hence it is transitive. Thus, it also contains the transitive closure of $\trl_c$, i.e.\ $\prec_c$, ending the proof.
\end{proof}

\begin{theorem}[Extending a correction function to a Pauli flow]\label{extending correction set to Pauli flow}
    Let $\Gamma = (G,I,O,\lambda)$ be a labelled open graph and let $c$ be a correction function on $\Gamma$. Then $c$ is (focused) extensive if and only if $(c, \prec_c)$ is a minimal (focused) Pauli flow.
\end{theorem}

\begin{proof}
    $(\Rightarrow)$: Suppose that $c$ is (focused) extensive. Then there exists $\prec$ on $\comp{O}$ such that $(c,\prec)$ is a (focused) Pauli flow. Thus, $(c,\prec_c)$ satisfies all of \ref{P4}-\ref{P9} (and \ref{F1}-\ref{F3}), as those conditions only ask about $c$. Further, $\prec_c$ contains $\trl_c$, and so $(c,\prec_c)$ satisfies \ref{P1}-\ref{P3}. Finally, $\prec$ is a strict partial order; hence $\prec_c$ is also a strict partial order, as it is a subset of $\prec$ by \ref{minimal order containment}. Therefore $(c,\prec_c)$ is a (focused) Pauli flow. It is minimal by Lemma~\ref{minimal order containment}.

    $(\Leftarrow)$: Immediate from the definitions.
\end{proof}

The above theorem tells us that, to verify whether a correction function $c$ is extensive, we only need to test one order induced by $c$. In other words, we only need to search for an extensive correction function rather than searching for a correction function and a partial order forming a flow together. The induced order is particularly well-behaved in the case of the focused correction function. Then, all Pauli-measured vertices are initial in the partial order, as explained in the following lemma.

\begin{lemma}[All Paulis are initial]\label{all pauli are first}
    Let $(G,I,O,\lambda)$ be a labelled open graph and $c$ be a focused correction function. Then $\forall u \in \measPauli . \forall v \in \comp{O} . \neg v \prec_c u$.
\end{lemma}

\begin{proof}
    Implied by \cite[Lemma B.11]{simmonsRelatingMeasurementPatterns2021}.
\end{proof}

Note that, in particular, the above lemma tells us the condition \ref{P3} never puts any requirements on~$\trl_c$ (and hence on $\prec_c$) if $c$ is a focused correction function.

\section{New algebraic interpretation}\label{sec:new algebraic interpretation}

As may be seen from the previous section, Pauli flow can be hard to work with, as the relevant definitions are long and complicated. For instance, to check if a given pair $(c,\prec)$ forms a focused Pauli flow, one needs to verify the twelve conditions \ref{P1}-\ref{P9} and \ref{F1}-\ref{F3} for each vertex, which is a tedious procedure.

We now provide an alternative interpretation of Pauli flow: Given a labelled open graph $\Gamma$, we define two matrices constructed from the adjacency matrix -- a `flow-demand matrix' $M_{\Gamma}$ and an `order-demand matrix' $N_{\Gamma}$. We prove that $\Gamma$ has Pauli flow if and only if there exists a matrix $C$ (encoding the correction function) such that $M_{\Gamma}C$ is the identity matrix and $N_{\Gamma}C$ forms a DAG:

\newcommand{\statealgebraicinterp}{%
Let $\Gamma = (G,I,O,\lambda)$ be a labelled open graph, $c$ be a correction function, $M$ be the flow-demand matrix of $\Gamma$, and $N$ be the order-demand matrix of $\Gamma$. Then $c$ is focused extensive if and only if the following two facts hold for the corresponding correction matrix $C$:\begin{itemize}
        \item $M_{\Gamma}C = Id_{\comp{O}}$,
        \item $N_{\Gamma}C$ is the adjacency matrix of a DAG.
    \end{itemize}
}

\begin{theorem}\label{algebraic interpretation section start}
 \statealgebraicinterp
\end{theorem}

The remainder of this section is dedicated to defining terms from the above theorem and to proving this theorem.

\subsection{Matrix definitions}

First, we define the relevant matrices. All matrices are over the field $\mathbb{F}_2$. We drop the subscripts specifying the ((labelled) open) graph when it is clear from context. We start with the standard definition of the adjacency matrix.

It will sometimes be useful to refer to individual rows or columns of a matrix, we will use the notation $M_{v,*}$ to denote the $v$-th row of $M$ and the notation $C_{*,u}$ to denote the $u$-th column of $C$.

\begin{definition}[Adjacency matrix]
    Let $G = (V,E)$ be a simple graph.
    Its \textit{adjacency matrix} is the $n \times n$ matrix $Adj_G$ with rows and columns corresponding to elements of $v$, where:\begin{equation*}
        \left(Adj_G\right)_{u,v} = \begin{cases}
            1 & \text{if } uv \in E\\
            0 & \text{otherwise}
        \end{cases}
    \end{equation*}
\end{definition}

A particularly important submatrix of the adjacency matrix of an open graph is the so-called `reduced adjacency matrix'.

\begin{definition}[Reduced adjacency matrix {\cite[Definition 6]{mhallaWhichGraphStates2014a}} and {\cite[Definition 3.11]{mitosekPauliFlowOpen2024}}]
    Let $(G,I,O)$ be an open graph. Its \textit{reduced adjacency matrix} is the $(n - n_O) \times (n - n_I)$ submatrix ${Adj_G}\mid_{\comp{O}}^{\comp{I}}$ of the adjacency matrix $Adj_G$, keeping only the rows corresponding to the non-outputs $\comp{O}$ and the columns corresponding to the non-inputs $\comp{I}$.
\end{definition}

The above matrix was previously used in characterising focused gflow on labelled open graphs with only $XY$ planar measurements (gflow being the straightforward restriction of Pauli flow to that setting) \cite{mhallaWhichGraphStates2014a}.
In particular, it was shown that a labelled open graph with all-$XY$ measurements has Pauli flow if and only if its reduced adjacency matrix is right-invertible, and one of the right inverses is the adjacency matrix of a DAG.
A generalisation of this result, which additionally allows Pauli-$X$ measurements, was stated by one of the authors in \cite{mitosekPauliFlowOpen2024}, deferring the proof to an upcoming joint paper: which is this work.
This previous paper by one of the authors also used an algebraic interpretation to deal with Pauli $X$ and $Z$ measurements by showing that a particular matrix called `flow matrix' must be right-invertible \cite{mitosekPauliFlowOpen2024}.
We combine and indeed improve on those results here, by defining what we now call a `flow-demand matrix', which works for all six measurement types.

\begin{definition}[Flow-demand matrix]\label{flow-demand matrix definition}
    Let $\Gamma = (G,I,O,\lambda)$ be a labelled open graph. We define the \textit{flow-demand matrix} $M_{\Gamma}$ as the $(n-n_O) \times (n-n_I)$ matrix with rows corresponding to non-outputs $\comp{O}$ and columns corresponding to non-inputs $\comp{I}$, where the row $M_{v,*}$ corresponding to the vertex $v \in \comp{O}$ satisfies the following for any $w \in \comp{I} \setminus \{ v \}$
    :\begin{itemize}
        \item if $\lambda(v) \in \{ X, XY \}$, then $M_{v,v} = 0$ and $M_{v,w} = Adj_{v,w}$ i.e.\ the $v$ row encodes the neighbourhood of $v$,
        \item if $\lambda(v) \in \{Z, YZ, XZ\}$, then $M_{v,v} = 1$ and $M_{v,w} = 0$, i.e.\ the $v$ row contains a $1$ at the intersection with the $v$ column and is identically $0$ otherwise, and
        \item if $\lambda(v) \in \{ Y \}$, then $M_{v,v} = 1$ (provided that $v$ column exists) and $M_{v,w} = Adj_{v,w}$, i.e.\ the $v$ row encodes the neighbourhood of $v$ and also has a $1$ at the intersection with the $v$ column.
    \end{itemize}
\end{definition}

As we will soon prove, the right invertibility of the flow-demand matrix relates to most, but not all, of the conditions of the focused Pauli flow. For the remaining conditions, we need the following construction we call the `order-demand matrix'.

\begin{definition}[Order-demand matrix]\label{order-demand matrix definition}
     Let $\Gamma = (G,I,O,\lambda)$ be a labelled open graph. We define the \textit{order-demand matrix} $N_{\Gamma}$ as the $(n-n_O) \times (n-n_I)$ matrix with rows corresponding to non-outputs $\comp{O}$ and columns corresponding to non-inputs $\comp{I}$, where the row $N_{v,*}$ corresponding to the vertex $v \in \comp{O}$ satisfies the following for any $w \in \comp{I}$:\begin{itemize}
        \item if $\lambda(v) \in \{ X, Y, Z \}$, then $N_{v,*} = \mathbf{0}$, i.e.\ the $v$ row is identically $0$,
        \item if $\lambda(v) = YZ$, then $N_{v,v} = 0$ and $N_{v,w} = Adj_{v,w}$, i.e.\ the $v$ row encodes the neighbourhood of $v$,
        \item if $\lambda(v) = XZ$, then $N_{v,v} = 1$ and $N_{v,w} = Adj_{v,w}$, i.e.\ the $v$ row encodes the neighbourhood of $v$ and also has a $1$ at the intersection with the $v$ column, and
        \item if $\lambda(v) = XY$, then $N_{v,v} = 1$ (provided that the $v$ column exists) and $N_{v,w} = 0$, i.e.\ the $v$ row contains a $1$ at the intersection with the $v$ column and is identically $0$ otherwise.
    \end{itemize}
\end{definition}

In order-demand matrices, the rows of $XY$ measured vertices are encoded as rows of $Z$ measured vertices would be in the flow-demand matrix. Similar correspondence holds between $YZ$ and $X$ measurements and between $XZ$ and $Y$ measurements. In the definitions, when setting the intersection of $v$ row and $v$ column to $1$, we must exclude $v \in I$, as for such $v$ there exists a row, but not a column. Note, that rows of Pauli measurements and $XY$ measured inputs are identically $0$ in the order-demand matrix. See Figure~\ref{fig:M} for the flow-demand matrix of the running example and Figure~\ref{fig:N} for the corresponding order-demand matrix. See Algorithm~\ref{M and N pseudo} in Appendix~\ref{pseudocodes easy} for the pseudocode description of the construction of the flow-demand and the order-demand matrices.

\begin{figure}
    \centering
    \begin{subfigure}[b]{0.43\textwidth}
        \centering
        $\begin{array}{|c|cccccc|}
            \toprule
            M & a & b & e & d & o_1 & o_2 \\
            \midrule
            i & 0 & 1 & 0 & 0 & 0 & 0 \\
            a & 1 & 0 & 0 & 0 & 0 & 0 \\
            b & 0 & 1 & 1 & 1 & 0 & 1 \\
            e & 0 & 1 & 0 & 1 & 1 & 1 \\
            d & 0 & 0 & 0 & 1 & 0 & 0 \\
            \bottomrule
        \end{array}$
        \caption{Flow-demand matrix of the graph in Figure~\ref{fig:loG}.}
        \label{fig:M}
    \end{subfigure}
    \qquad
    \begin{subfigure}[b]{0.43\textwidth}
        \centering
        $\begin{array}{|c|cccccc|}
            \toprule
            N & a & b & e & d & o_1 & o_2 \\
            \midrule
            i & 0 & 0 & 0 & 0 & 0 & 0 \\
            a & 1 & 0 & 0 & 1 & 0 & 1 \\
            b & 0 & 0 & 0 & 0 & 0 & 0 \\
            e & 0 & 0 & 1 & 0 & 0 & 0 \\
            d & 0 & 0 & 0 & 0 & 0 & 0 \\
            \bottomrule
        \end{array}$
        \caption{Order-demand matrix of the graph in Figure~\ref{fig:loG}.}
        \label{fig:N}
    \end{subfigure}
    \begin{subfigure}[b]{0.43\textwidth}
        \centering
        $\begin{array}{|c|ccccc|}
            \toprule
            C & i & a & b & e & d \\
            \midrule
            a & 0 & 1 & 0 & 0 & 0 \\
            b & 1 & 0 & 0 & 0 & 0 \\
            e & 1 & 1 & 1 & 0 & 0 \\
            d & 0 & 0 & 0 & 0 & 1 \\
            o_1 & 1 & 1 & 0 & 1 & 0 \\
            o_2 & 0 & 1 & 0 & 0 & 1 \\
            \bottomrule
        \end{array}$
        \caption{Matrix of the correction function in Figure~\ref{fig:pf} for the labelled open graph in Figure~\ref{fig:loG}.}
        \label{fig:C}
    \end{subfigure}
    \qquad
    \begin{subfigure}[b]{0.43\textwidth}
        \centering
        $\begin{array}{|c|ccccc|}
            \toprule
            NC & i & a & b & e & d \\
            \midrule
            i & 0 & 0 & 0 & 0 & 0 \\
            a & 0 & 0 & 0 & 0 & 0 \\
            b & 0 & 0 & 0 & 0 & 0 \\
            e & 1 & 1 & 1 & 0 & 0 \\
            d & 0 & 0 & 0 & 0 & 0 \\
            \bottomrule
        \end{array}$
        \caption{Product of the order-demand matrix and the correction matrix in this figure.}
        \label{fig:NC}
    \end{subfigure}
    \caption{Various matrices encoding aspects of Pauli flow on the labelled open graph in Figure~\ref{fig:loG}.}
    \label{fig:mats}
\end{figure}

Instead of thinking about the correction function as a function, we can transform it, too, into a matrix.

\begin{definition}[Correction matrix]\label{def:correction-matrix}
    Let $\Gamma = (G,I,O,\lambda)$ be a labelled open graph and let $c$ be a correction function on $\Gamma$. We define the \textit{correction matrix $C$} encoding the function $c$ as the $(n-n_I) \times (n-n_O)$ matrix with rows corresponding to non-inputs $\comp{I}$ and columns corresponding to non-outputs $\comp{O}$, where $C_{u,v} = 1$ if and only if $u \in c(v)$.
\end{definition}

\begin{observation}
 Given a correction matrix $C$ encoding an unknown correction function $c$, the correction function can be recovered as $c(v) = \{ u \in \comp{I} \mid C_{u,v} = 1 \}$.
\end{observation}

In other words, the $v$ column in $C$ encodes the characteristic function of set $c(v)$. That is $C_{*,v} = \ind{c(v)}{\comp{I}}$ and $\supp{C_{*,v}} = c(v)$. Here, $\ind{\someset}{\otherset}$ stands for the indicator function of the subset $\someset$ of $\otherset$ (as well as for the corresponding vector over $\mathbb{F}_2$), and $\support$ stands for the support. See Figure~\ref{fig:C} for a correction matrix corresponding to the running example.

\subsection{Properties of the flow-related matrices}

In this subsection, we prove useful properties of the newly-defined flow-demand and order-demand matrices. Some of these properties require us to split the conditions \ref{P4}, \ref{P5}, and \ref{P6} of the Pauli flow definition: They effectively have two parts each and it will be easier to consider each part on its own.

\begin{observation}\label{P456 split}
    Let $(G,I,O,\lambda)$ be a labelled open graph and let $c$ be a correction function. For all $u \in \comp{O}$, define the following conditions:
    \begin{enumerate}
        \item[{\crtcrossreflabel{(P4a)}[P4a]}] $\lambda(u) = XY \Rightarrow u \in \odd{c(u)}$
        \item[{\crtcrossreflabel{(P4b)}[P4b]}] $\lambda(u) = XY \Rightarrow u \notin c(u)$
        \item[{\crtcrossreflabel{(P5a)}[P5a]}] $\lambda(u) = XZ \Rightarrow u \in c(u)$
        \item[{\crtcrossreflabel{(P5b)}[P5b]}] $\lambda(u) = XZ \Rightarrow u \notin \codd{c(u)}$
        \item[{\crtcrossreflabel{(P6a)}[P6a]}] $\lambda(u) = YZ \Rightarrow u \in c(u)$
        \item[{\crtcrossreflabel{(P6b)}[P6b]}] $\lambda(u) = YZ \Rightarrow u \notin \odd{c(u)}$.
    \end{enumerate}
    Then, \ref{P4} is equivalent to the conjunction of \ref{P4a} and \ref{P4b}; \ref{P5} is equivalent to the conjunction of \ref{P5a} and \ref{P5b}; and \ref{P6} is equivalent to the conjunction of \ref{P6a} and \ref{P6b}.

    In particular, note that if $c$ satisfies \ref{P5}, then $u \in c(u)$ and $u \in \odd{c(u)}$, so $u \notin c(u) \symd \odd{c(u)} = \codd{u}$ and thus \ref{P5a} and \ref{P5b} follow. When $c$ satisfies \ref{P5a} and \ref{P5b}, then $u \in c(u)$ and $u \notin \codd{c(u)} = c(u) \symd \odd{c(u)}$ and hence $u \in \odd{c(u)}$, i.e.\ \ref{P5} follows.
\end{observation}

Now, we prove different facts regarding matrix products involving flow-demand or order-demand matrices. First, we look at the product of the flow-demand matrix and a column vector.

\begin{lemma}\label{row by col meaning}
    Let $(G,I,O,\lambda)$ be a labelled open graph. Let $\someset \subseteq \comp{I}$. Let $u \in \comp{I}$. Then the product of the $u$-th row of $M$ with the indicator vector of $\someset$ satisfies $M_{u,*}\ind{\someset}{\comp{I}} = \left(M\ind{\someset}{\comp{I}}\right)_u = 1$ if and only if:\begin{align*}
        \triplecase
            {\lambda(u) \in \{ X, XY \}}{u \in \odd{\someset}}
            {\lambda(u) \in \{ XZ, YZ, Z \}}{u \in \someset}
            {\lambda(u) = Y}{u \in \codd{\someset}}
    \end{align*}
\end{lemma}

\begin{proof}
    The fact $M_{u,*}\ind{\someset}{\comp{I}} = \left(M\ind{\someset}{\comp{I}}\right)_u$ follows immediately from the definition of matrix multiplication. Now, we consider different cases depending on $\lambda(u)$. Here, $w \in \comp{I}$ is any non-input.\begin{enumerate}
        \item Suppose $\lambda(u) \in \{ X, XY \}$. Then $M_{u,*}$ encodes $Adj(u)$ and thus:\begin{align*}
            M_{u,w}\left(\ind{\someset}{\comp{I}}\right)_w = \begin{cases}
                1 & \text{if } uw \in E \wedge w \in \someset\\
                0 & \text{otherwise}
            \end{cases}
        \end{align*}
        Recall that the matrices and computations are over $\mathbb{F}_2$.
        Then, we have:\begin{align*}
            \left(M\ind{\someset}{\comp{I}}\right)_u
            = \sum_{w \in \comp{I}} M_{u,w}\left(\ind{\someset}{\comp{I}}\right)_w
            &= \abs{\left\{w\in \comp{I} : uw \in E \wedge w \in \someset\right\}} \bmod{2} \\
            &= \abs{\left\{w\in \someset : uw \in E\right\}} \bmod{2}
            = \begin{cases}
                1 &\text{if } u \in \odd{\someset}\\
                0 &\text{otherwise}
            \end{cases}
        \end{align*}

        \item Suppose $\lambda(u) \in \{ XZ, YZ, Z \}$. Then $M_{u,*}$ equals $0$ except for $M_{u,u} = 1$ and thus:\begin{align*}
            M_{u,w}\left(\ind{\someset}{\comp{I}}\right)_w = \begin{cases}
                1 & \text{if } u = w \wedge w \in \someset\\
                0 & \text{otherwise}
            \end{cases}
        \end{align*}
        Note, that $u \in \comp{I}$ because inputs cannot be measured in $XZ, YZ, Z$ (cf.\ Definition~\ref{def:labelling}). Therefore:\begin{align*}
            \left(M\ind{\someset}{\comp{I}}\right)_u = \sum_{w \in \comp{I}} M_{u,w}\left(\ind{\someset}{\comp{I}}\right)_w = M_{u,u}\left(\ind{\someset}{\comp{I}}\right)_u = \begin{cases}
                1 & \text{if } u \in \someset\\
                0 & \text{otherwise}
            \end{cases}
        \end{align*}

        \item Suppose $\lambda(u) = Y$. Then $M_{u,*}$ encodes $Adj(u)$ except for $M_{u,u} = 1$ and thus:\begin{align*}
            M_{u,w}\left(\ind{\someset}{\comp{I}}\right)_w = \begin{cases}
                1 & \text{if } uw \in E \wedge w \in \someset\\
                1 & \text{if } u = w \wedge w \in \someset\\
                0 & \text{otherwise}
            \end{cases}
        \end{align*}
        The first two options are disjoint as the graph $G$ is simple, so it does not contain the edge $uu$. Therefore, the expression for $\left(M\ind{\someset}{\comp{I}}\right)_u$ is the sum of the expressions obtained in the previous two cases. Hence there are the following cases for $\left(M\ind{\someset}{\comp{I}}\right)_u = \sum_{w\in\comp{I}} M_{u,w}\left(\ind{\someset}{\comp{I}}\right)_w$:\begin{itemize}
            \item When $u \in \odd{\someset}$ and $u \in \someset$, we get $\left(M\ind{\someset}{\comp{I}}\right)_u = 1+1 = 0$ (recall that we work over $\mathbb{F}_2$),
            \item When $u \notin \odd{\someset}$ and $u \in \someset$, we get $\left(M\ind{\someset}{\comp{I}}\right)_u = 0+1 = 1$,
            \item When $u \in \odd{\someset}$ and $u \notin \someset$, we get $\left(M\ind{\someset}{\comp{I}}\right)_u = 1+0 = 1$,
            \item When $u \notin \odd{\someset}$ and $u \notin \someset$, we get $\left(M\ind{\someset}{\comp{I}}\right)_u = 0+0 = 0$.
        \end{itemize}
        Combining the conditions, we obtain:
        \begin{equation*}
            \left(M\ind{\someset}{\comp{I}}\right)_u = \begin{cases}
                1 & \text{if } u \in \codd{\someset}\\
                0 & \text{otherwise}
            \end{cases}
        \end{equation*}
    which ends the proof.\qedhere
    \end{enumerate}
\end{proof}

Using the above lemma, we can provide an interpretation of the product of the flow-demand matrix and the correction matrix.

\begin{lemma}\label{row meaning}
    Let $(G,I,O,\lambda)$ be a labelled open graph and let $c$ be a correction function. Let $u,v \in \comp{O}$.
    Then, $(MC)_{u,v} = 1$ if and only if:
    \begin{align*}
        \triplecase
            {\lambda(u) \in \{X, XY\}}{u \in \odd{c(v)}}
            {\lambda(u) \in \{XZ, YZ, Z \}}{u \in c(v)}
            {\lambda(u)=Y}{u \in \codd{c(v)}}
        \intertext{Further, $(NC)_{u,v} = 1$ if and only if:}
        \triplecase
            {\lambda(u) = YZ}{u \in \odd{c(v)}}
            {\lambda(u) = XY}{u \in c(v)}
            {\lambda(u) = XZ}{u \in \codd{c(v)}}
    \end{align*}
\end{lemma}

\begin{proof}
    We start with $(MC)_{u,v} = M_{u,*} C_{*,v}$. Columns of $C$ encode $c(v)$, i.e.\ $C_{*,v} = \ind{c(v)}{\comp{I}}$. Hence, from Lemma~\ref{row by col meaning}, we have that $M_{u,*} C_{*,v}$ equals $1$ if and only if:\begin{align*}
        \triplecase
            {\lambda(v) \in \{X, XY\}}{v \in \odd{c(v)}}
            {\lambda(v) \in \{XZ, YZ, Z \}}{v \in c(v)}
            {\lambda(v)=Y}{v \in \codd{c(v)}}
    \end{align*}
    The proof for $(NC)_{u,v}$ is analogous, taking into account the construction of each row of $N$ as given in Definition~\ref{order-demand matrix definition}.
\end{proof}

We can now relate several of the Pauli flow conditions from Definition~\ref{def:Pauli-flow} and Observation~\ref{P456 split}, as well as the focusing conditions from Definition~\ref{def:focused}, to the matrix product $MC$.

\begin{theorem}\label{MC}
    Let $\Gamma = (G,I,O,\lambda)$ be a labelled open graph and let $c$ be a correction function on $\Gamma$. The following two statements are equivalent:\begin{itemize}
        \item $MC = Id_{\comp{O}}$, where $C$ is the correction matrix encoding $c$ according to Definition~\ref{def:correction-matrix} and $M$ is the flow-demand matrix of $\Gamma$,
        \item $c$ satisfies \ref{P4a}, \ref{P5a}, \ref{P6a}, \ref{P7}, \ref{P8}, \ref{P9}, \ref{F1}, \ref{F2}, and \ref{F3}.
    \end{itemize}
\end{theorem}

\begin{proof}
    $(\Leftarrow)$: We must show that:\begin{equation*}
        (MC)_{u,v} = \begin{cases}
            1 & \text{if } u = v\\
            0 & \text{otherwise}
        \end{cases}
    \end{equation*}
    We consider three cases depending on $\lambda(u)$.
    \begin{enumerate}
        \item $\lambda(u) \in \{ X, XY \}$. By Lemma~\ref{row meaning}:\begin{equation*}
            (MC)_{u,v} = \begin{cases}
                1 & \text{if } u \in \odd{c(v)}\\
                0 & \text{otherwise}
            \end{cases}.
        \end{equation*}
        Now, we have $u \in \odd{c(u)}$ by \ref{P4a} for $\ld(u)=XY$ and by \ref{P7} for $\lambda(u) = X$. Furthermore, $u \notin \odd{c(v)}$ for $u \ne v$ by \ref{F2} for either label.
        Thus, $(MC)_{u,v}=1$ if and only if $u=v$.
    
        \item $\lambda(u) \in \{ XZ, YZ, Z \}$. By Lemma~\ref{row meaning}:\begin{equation*}
            (MC)_{u,v} = \begin{cases}
                1 & \text{if } u \in c(v)\\
                0 & \text{otherwise}
            \end{cases}.
        \end{equation*}
        Now, we have $u \in c(u)$ by \ref{P5a} for $\ld(u)=XZ$, by \ref{P6a} for $\ld(u)=YZ$, and by \ref{P8} for $\ld(u)=X$.
        Furthermore, $u \notin c(v)$ for $u \ne v$ by \ref{F1}.
        Thus, $(MC)_{u,v}=1$ if and only if $u=v$.

        \item $\lambda(u) = Y$. By Lemma~\ref{row meaning}:\begin{equation*}
            (MC)_{u,v} = \begin{cases}
                1 & \text{if } u \in \codd{c(v)}\\
                0 & \text{otherwise}
            \end{cases}.
        \end{equation*}
        Now, we have $u \in \codd{c(u)}$ by \ref{P9}. Finally, $u \notin \codd{c(v)}$ for $v \ne u$ by \ref{F3}.
        Thus, again, $(MC)_{u,v}=1$ if and only if $u=v$.
    \end{enumerate}
    Hence, indeed we have \begin{equation*}
        (MC)_{u,v} = \begin{cases}
            1 & \text{if } u = v\\
            0 & \text{otherwise}
        \end{cases}
    \end{equation*}
    for any $u,v \in \comp{O}$, which ends the proof.

    $(\Rightarrow)$: This direction is very similar. By assumption, $MC = Id_{\comp{O}}$, so:\begin{equation}\label{id meaning}
        (MC)_{u,v} = \begin{cases}
            1 & \text{if } u = v \\
            0 & \text{otherwise}
        \end{cases}
    \end{equation}
    Again, we consider three cases depending on $\lambda(u)$.\begin{enumerate}
        \item $\lambda(u) \in \{X, XY \}$. By Lemma~\ref{row meaning}, we get:\begin{equation*}
            (MC)_{u,v} = \begin{cases}
                1 & \text{if } u \in \odd{c(v)} \\
                0 & \text{otherwise}
            \end{cases}
        \end{equation*}
        Combining it with equation \ref{id meaning}, we get $u \in \odd{c(u)}$ for any $u$ measured in $X, XY$ and thus \ref{P4a}, \ref{P7} hold. Further, we get that $u \notin \odd{c(v)}$ when $u \ne v$ for any $u$ measured in $X, XY$ and thus $\odd{c(v)}$ may only contain $v$ and vertices measured in $\{ XZ, YZ, Y, Z \}$, i.e.\ \ref{F2} holds.

        \item $\lambda(u) \in \{XZ, YZ, Z\}$. By Lemma~\ref{row meaning}, we get:\begin{equation*}
            (MC)_{u,v} = \begin{cases}
                1 & \text{if } u \in c(v) \\
                0 & \text{otherwise}
            \end{cases}
        \end{equation*}
        Combining it with equation \ref{id meaning}, we get $u \in c(u)$ for any $u$ measured in $XZ, YZ, Z$ and thus \ref{P5a}, \ref{P6a}, \ref{P8} hold. Further, we get that $u \notin c(v)$ when $u \ne v$ for any $u$ measured in $XZ, YZ, Z$ and thus $c(v)$ may only contain $v$ and vertices measured in $\{ X, XY, Y \}$, i.e.\ \ref{F1} holds.

        \item $\lambda(u) = Y$. By Lemma~\ref{row meaning}, we get:\begin{equation*}
            (MC)_{u,v} = \begin{cases}
                1 & \text{if } u \in \codd{c(v)} \\
                0 & \text{otherwise}
            \end{cases}
        \end{equation*}
       Combining it with equation \ref{id meaning}, we get $u \in \codd{c(u)}$ for all $u$ measured in $Y$ and thus \ref{P9} holds. Further, we get that $u \notin \codd{c(v)}$ when $u \ne v$ for all $u$ measured in $Y$, i.e.\ \ref{F3} holds.\qedhere
    \end{enumerate}
\end{proof}

Similarly, the product $NC$ can be related to the remaining condition from the flow definition. First, we show that $NC$ without its main diagonal encodes precisely the induced relation $\trl_c$.

\begin{lemma}\label{NC is the induced order}
    Let $\Gamma = (G,I,O,\lambda)$ be a labelled open graph and let $c$ be a focused correction function on $\Gamma$. Then, for any $u, v \in \comp{O}$ such that $u \ne v$, we have $(NC)_{u,v} = 1$ if and only if $v \trl_c u$. In other words, the off-diagonal part of the matrix $NC$ encodes the relation $\trl_c$.
\end{lemma}

\begin{proof}
    Let $u,v \in \comp{O}$ such that $u \ne v$. By Lemma~\ref{all pauli are first}, if $v \trl_c u$ then $\lambda(u) \in \{ XY, YZ, XZ \}$. Now, we consider three cases on $\lambda(u)$.\begin{enumerate}
        \item $\lambda(u) = XY$. Then, by Lemma~\ref{row meaning}:\begin{equation*}
            (NC)_{u,v} = \begin{cases}
                1 & \text{if } u \in c(v)\\
                0 & \text{otherwise}
            \end{cases}
        \end{equation*}
        Since $u \ne v$ and $\lambda(u) \notin \{ X, Y \}$, by \ref{P1} we get $v \trl_c u$ if and only if $u \in c(v)$, as desired.
        
        \item $\lambda(u) = YZ$. Then, by Lemma~\ref{row meaning}:\begin{equation*}
            (NC)_{u,v} = \begin{cases}
                1 & \text{if } u \in \odd{c(v)}\\
                0 & \text{otherwise}
            \end{cases}
        \end{equation*}
        Since $u \ne v$ and $\lambda(u) \notin \{ Y, Z \}$, by \ref{P2} we get $v \trl_c u$ if and only if $u \in \odd{c(v)}$, as desired.

        \item $\lambda(u) = XZ$. Then, by Lemma~\ref{row meaning}:\begin{equation*}
            (NC)_{u,v} = \begin{cases}
                1 & \text{if } u \in \codd{c(v)}\\
                0 & \text{otherwise}
            \end{cases}
        \end{equation*}
        Since $u \ne v$ and $\lambda(u) \notin \{ Y, Z \}$, by \ref{P2} we get $v \trl_c u$ if and only if $u \in \odd{c(v)}$. As $c$ is focused, it satisfies \ref{F1} and thus $u \notin c(v)$. Hence $u \in \codd{c(v)} \Leftrightarrow u \in \odd{c(v)}$. Combining with the above, we again get the desired equivalence. \qedhere
    \end{enumerate}
\end{proof}

Now, we can relate the matrix product $NC$ to the Pauli flow conditions.

\begin{theorem}\label{NC}
    Let $\Gamma = (G,I,O,\lambda)$ be a labelled open graph and let $c$ be a focused correction function on $\Gamma$. The following two statements are equivalent:\begin{itemize}
        \item $NC$ is the adjacency matrix of a DAG, where $C$ is the correction matrix encoding $c$ according to Definition~\ref{def:correction-matrix} and $N$ is the order-demand matrix of $\Gamma$,
        \item $\prec_c$ is a partial order on $\comp{O}$ and the tuple $(c, \prec_c)$ satisfies \ref{P1}, \ref{P2}, \ref{P3}, \ref{P4b}, \ref{P5b}, and \ref{P6b}.
    \end{itemize}
\end{theorem}

\begin{proof}
    $(\Rightarrow)$: By Lemma~\ref{NC is the induced order}, $\trl_c$ is encoded entirely within $NC$. Since $NC$ is a DAG, $\trl_c$ is an acyclic relation and thus its transitive closure is a partial order. Hence, $\prec_c$ satisfies \ref{P1}, \ref{P2}, \ref{P3} and is a partial order. What remains is to prove \ref{P4b}, \ref{P5b}, \ref{P6b}. Let $u \in \comp{O}$ be any planar measured vertex. Since $NC$ is a DAG, it cannot contain a one-cycle (i.e.\ a loop) for any vertex $u$, and thus:\begin{equation}\label{main diag of NC}
        (NC)_{u,u} = 0
    \end{equation}
    We consider three cases depending on $\ld(u)$.\begin{enumerate}
        \item $\lambda(u) = XY$. By Lemma~\ref{row meaning}:\begin{equation*}
            (NC)_{u,u} = \begin{cases}
                1 & \text{if } u \in c(u)\\
                0 & \text{otherwise}
            \end{cases}
        \end{equation*}
        Combining this with \ref{main diag of NC}, we find $u \notin c(u)$ for any $u$ measured in $XY$, i.e.\ \ref{P4b} holds.

        \item $\lambda(u) = XZ$. By Lemma~\ref{row meaning}:\begin{equation*}
            (NC)_{u,u} = \begin{cases}
                1 & \text{if } u \in \codd{c(u)}\\
                0 & \text{otherwise}
            \end{cases}
        \end{equation*}
        Combining this with \ref{main diag of NC}, we find $u \notin \codd{c(u)}$ for any $u$ measured in $XZ$, i.e.\ \ref{P5b} holds.

        \item $\lambda(u) = YZ$. By Lemma~\ref{row meaning}:\begin{equation*}
            (NC)_{u,u} = \begin{cases}
                1 & \text{if } u \in \odd{c(u)}\\
                0 & \text{otherwise}
            \end{cases}
        \end{equation*}
        Combining this with \ref{main diag of NC}, we find $u \notin \odd{c(u)}$ for any $u$ measured in $YZ$, i.e.\ \ref{P6b} holds.
    \end{enumerate}

    $(\Leftarrow)$: Similarly to the other direction, from \ref{P4b}, \ref{P5b}, and \ref{P6b}, we get $(NC)_{u,u} = 0$ for all planar measured $u$. Further, the off-diagonal entries of $NC$ encode $\trl_c$ by Lemma~\ref{NC is the induced order}. Hence, as $\prec_c$ is a partial order, $\trl_c$ is acyclic and thus the off-diagonal part of $NC$ is acyclic. Combining the two facts, we indeed find that $NC$ is the adjacency matrix of a DAG.
\end{proof}

Combining the previous theorems, we obtain the desired algebraic interpretation of Pauli flow: the main result of this paper.

\begin{repthm}{algebraic interpretation section start}
 \statealgebraicinterp
\end{repthm}

\begin{proof}
    $(\Rightarrow)$: Let $c$ be focused extensive. By Theorem~\ref{extending correction set to Pauli flow}, this means $(c,\prec_c)$ is a minimal focused Pauli flow. Then by Theorem~\ref{MC}, $MC = Id_{\comp{O}}$. Moreover, by Theorem~\ref{NC}, $NC$ is the adjacency matrix of a DAG.

    $(\Leftarrow)$: This follows immediately from Theorems~\ref{MC} and~\ref{NC}, as well as Observation~\ref{P456 split}.
\end{proof}

As an example, one can verify that the product of the flow-demand matrix in Figure~\ref{fig:M} and the correction matrix in Figure~\ref{fig:C} is an identity.
Moreover, as shown in Figure~\ref{fig:NC}, the product of the order-demand matrix in Figure~\ref{fig:N} and the previous correction matrix is a DAG (in fact, it is the same DAG as the one in Figure~\ref{fig:pf}).

When the numbers of inputs and outputs are equal, i.e.\ $n_I=n_O$, the flow-demand matrix is square. Thus, it has a unique right inverse (if any), which is the inverse. By Theorem~\ref{algebraic interpretation section start}, to verify the existence of flow we only need to check that such inverse exists and that the product of the order-demand matrix and the inverse of the flow-demand matrix forms a DAG. For an example of this $n_I=n_O$ case, see Figure~\ref{fig:loGnoflow}.

The case of equal numbers of inputs and outputs has another interesting aspect -- it is possible to switch the sets of inputs and outputs, obtaining another labelled open graph which may have Pauli flow (if, on the other hand, the numbers of inputs and outputs are different, then it is immediate that at most one of the pair of labelled open graphs could have flow).
Indeed, this reversal transformation turns out to preserve Pauli flow, as we show in the next subsection.

\begin{figure}
    \centering
    \begin{subfigure}[b]{0.20\textwidth}
        \begin{equation*}
            \begin{tikzpicture}
	\begin{pgfonlayer}{nodelayer}
		\node [style=Small Empty Box] (0) at (-2, -1) {};
		\node [style=Small Black] (1) at (-2, -1) {};
		\node [style=Small Empty] (2) at (2, -1) {};
		\node [style=Small Black] (3) at (0, 1) {};
		\node [style=none] (4) at (-2.5, -1) {$i$};
		\node [style=none] (5) at (2.5, -1) {$o$};
		\node [style=none] (6) at (0, 1.5) {$a$};
		\node [style=none] (7) at (3, 0) {};
		\node [style=none] (8) at (-3, 0) {};
		\node [style=none] (9) at (0, 2) {};
		\node [style=none] (10) at (0, -2) {};
	\end{pgfonlayer}
	\begin{pgfonlayer}{edgelayer}
		\draw (0) to (3);
		\draw (3) to (2);
		\draw (2) to (0);
	\end{pgfonlayer}
\end{tikzpicture}
        \end{equation*}
    \end{subfigure}
    \begin{subfigure}[b]{0.14\textwidth}
        $\begin{array}{|c|c|}
            \toprule
            v & \lambda(v) \\
            \midrule
            i & X \\
            a & YZ \\
            \bottomrule
        \end{array}$
    \end{subfigure}
    \hfill
    \begin{subfigure}[b]{0.14\textwidth}
        $\begin{array}{|c|c|c|}
            \toprule
            M & a & o \\
            \midrule
            i & 1 & 1 \\
            a & 1 & 0 \\
            \bottomrule
        \end{array}$
    \end{subfigure}
    \begin{subfigure}[b]{0.14\textwidth}
        $\begin{array}{|c|c|c|}
            \toprule
            N & a & o \\
            \midrule
            i & 0 & 0 \\
            a & 0 & 1 \\
            \bottomrule
        \end{array}$
    \end{subfigure}
    \begin{subfigure}[b]{0.14\textwidth}
        $\begin{array}{|c|c|c|}
            \toprule
            C & i & a \\
            \midrule
            a & 0 & 1 \\
            o & 1 & 1 \\
            \bottomrule
        \end{array}$
    \end{subfigure}
    \hfill
    \begin{subfigure}[b]{0.15\textwidth}
        $\begin{array}{|c|c|c|}
            \toprule
            NC & i & a \\
            \midrule
            i & 0 & 0 \\
            a & 1 & 1 \\
            \bottomrule
        \end{array}$
    \end{subfigure}
    \caption{Example of a labelled open graph with one input $i$ and one output $o$, together with its flow-demand matrix $M$, order-demand matrix $N$, and the unique inverse $C$ of $M$. The product $NC$ contains a loop at $a$ and is therefore not a DAG, so the labelled open graph has no flow.}
    \label{fig:loGnoflow}
\end{figure}
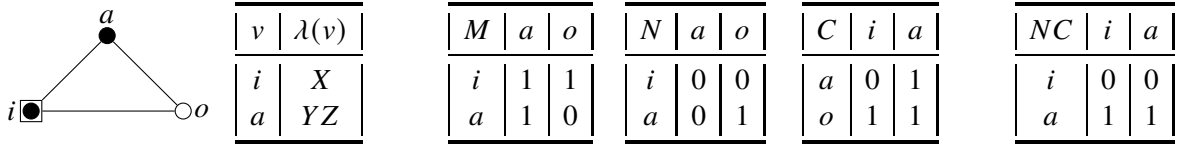

\subsection{Reversibility of Pauli flow}

For labelled open graphs where the number of inputs is equal to the number of outputs and all measurements are of $XY$ type, the previous algebraic interpretation was used to show two properties of focused gflow: it is unique, and it is reversible \cite[Theorem~4]{mhallaWhichGraphStates2014a}.
Here, reversing a labelled open graph means swapping the roles of inputs and outputs.
Reversibility of the flow then means that the the reverse \LOG\ has gflow if and only if the original \LOG\ has gflow.
Moreover, the correction matrix of the gflow on the reverse \LOG{} was shown to be $C^T$, where $C$ is the correction matrix of the original gflow.

Using the new algebraic interpretation for all measurement types, it has been straightforward to show the uniqueness of focused Pauli flow on \LOG{}s with equal numbers of inputs and outputs.
We now prove that under these conditions, focused Pauli flow can also be reversed, although the relationship between the reversed and the original flow is less straightforward if there are measurements of type $XZ,YZ$, or $Z$.

We will assume throughout this subsection that all labelled open graphs satisfy $I \cap O = \emptyset$.
By \cite[Theorem~3.10]{mitosekPauliFlowOpen2024}, this is without loss of generality regarding the existence of Pauli flow.

\begin{definition}
    Let $\Gamma=(G,I,O,\ld)$ be a labelled open graph.
    Define the following subsets of vertices:
    \begin{itemize}
        \item $\internal:=V\setminus(I\cup O)$, the set of internal vertices,
        \item $\Xlike := \{v\in \internal\mid \ld(v)\in\{XY,X,Y\}\}$, the set of `$X$-like' measured internal vertices,
        \item $\Zlike := \{v\in \internal\mid \ld(v)\in\{XZ,YZ,Z\}\}$, the set of `$Z$-like' measured internal vertices,
        \item $\paulis := \{v\in \internal\mid \ld(v)\in\{X,Y,Z\}\}$ the set of Pauli measured internal vertices, and
        \item $\planar := \{v\in \internal\mid \ld(v)\in\{XY,XZ,YZ\}\}$ the set of planar measured internal vertices.
    \end{itemize}
\end{definition}

\begin{definition}
    Let $\Gamma=(G,I,O,\ld)$ be a labelled open graph which satisfies $\abs{I}=\abs{O}$.
    Then its \emph{reverse labelled open graph} is $\Gamma'=(G,O,I,\ld')$ where the role of inputs and outputs is swapped, and
    \begin{equation*}
        \ld'(v) =
        \begin{cases}
        \ld(v) &\text{if } v\in \internal \\
        XY &\text{if } v\in O.
        \end{cases}
    \end{equation*}
\end{definition}

\begin{definition}
    Let $\Gamma=(G,I,O,\ld)$ be a labelled open graph.
    Define the \emph{extended adjacency matrix} $\A$ of $\Gamma$ as
    \begin{equation*}
        \A_{vw}
        := \begin{cases}
          1 &\text{if } \{v,w\}\in E \vee (v=w \wedge \ld(v)\in\{Y,XZ\}) \\
          0 &\text{otherwise.}
         \end{cases}
    \end{equation*}
    We will need four disjoint submatrices of the extended adjacency matrix (which do not quite cover the full matrix):
    \begin{itemize}
        \item Let $\Azz$ be the submatrix of $\A$ whose rows are given by $I\cup\Xlike$ and whose columns are given by $\Xlike\cup O$.
        \item Let $\Azo$ be the submatrix of $\A$ whose rows are given by $I\cup\Xlike$ and whose columns are given by $\Zlike$.
        \item Let $\Aoz$ be the submatrix of $\A$ whose rows are given by $\Zlike$ and whose columns are given by $\Xlike\cup O$.
        \item Let $\Aoo$ be the submatrix of $\A$ whose rows and columns are both given by $\Zlike$.
    \end{itemize}
    In summary, where $*$ indicates unnamed blocks of the extended adjacency matrix:
    \begin{equation}\label{eq:extended-adjacency}
        \A =
        \begin{pNiceMatrix}[first-row, last-col]
            I & \Xlike\cup O & \Zlike & \\
            * & \Azz        & \Azo   & I\cup\Xlike \\
            * & *           & *      & O \\
            * & \Aoz        & \Aoo   & \Zlike
        \end{pNiceMatrix}
    \end{equation}
\end{definition}

\begin{observation}\label{obs:flow-demand-block}
    Let $\Gamma=(G,I,O,\ld)$ be a labelled open graph which satisfies $\abs{I}=\abs{O}$, and let $\Gamma'$ be its reverse.
    Then the flow demand matrices $M$ of $\Gamma$ and $M'$ of $\Gamma'$ are given by
    \begin{equation*}
        M = \begin{pNiceMatrix}[first-row, last-col]
            \Xlike\cup O & \Zlike & \\
            \Azz        & \Azo   & I\cup\Xlike \\
            0           & \I     & \Zlike
          \end{pNiceMatrix}
        \qquad\text{and}\qquad
        M' = \begin{pNiceMatrix}[first-row, last-col]
            I\cup\Xlike & \Zlike   & \\
            (\Azz)^T    & (\Aoz)^T & \Xlike\cup O \\
            0           & \I       & \Zlike
           \end{pNiceMatrix}
        = \begin{pmatrix} \Azz & 0 \\ \Aoz & \I \end{pmatrix}^T
    \end{equation*}
\end{observation}

\begin{observation}\label{obs:order-demand-block}
    Let $\Gamma=(G,I,O,\ld)$ be a labelled open graph which satisfies $\abs{I}=\abs{O}$, and let $\Gamma'$ be its reverse.
    The flow demand matrix $N$ of $\Gamma$ is a matrix whose rows are labelled by non-outputs, whose columns are labelled by non-inputs, and where the rows corresponding to inputs or Pauli-measured vertices are identically 0.
    We can thus write $N = P \fullN$, where $P$ is a diagonal matrix whose rows and columns are both labelled by non-outputs, such that:
    \begin{equation*}
        P_{vw}
        = \begin{cases}
         \dl_{vw} &\text{if } v\in\planar \\
         0 &\text{otherwise,}
        \end{cases}
        \qquad\text{and}\qquad
        \fullN
        = \begin{pNiceMatrix}[first-row, last-col]
         \Xlike\cup O & \Zlike & \\
         J           & 0      & I\cup\Xlike \\
         \Aoz        & \Aoo   & \Zlike
        \end{pNiceMatrix}
    \end{equation*}
    with $J_{vw}=\dl_{vw}$.
    Note that, while $J$ is square and takes values in $\{0,1\}$, it is not an identity matrix because the set of row labels is not the same as the set of column labels.
    
    Similarly, $N'= P'\fullN'$, where $P'$ is a diagonal matrix whose rows and columns are both labelled by non-inputs, such that:
    \begin{equation*}
        P'_{vw}
        = \begin{cases}
         \dl_{vw} &\text{if } v\in\planar \\
         0 &\text{otherwise,}
        \end{cases}
        \qquad\text{and}\qquad
        \fullN'
        = \begin{pNiceMatrix}[first-row, last-col]
         I\cup\Xlike & \Zlike & \\
         J^T         & 0      & \Xlike\cup O \\
         (\Azo)^T    & \Aoo   & \Zlike
        \end{pNiceMatrix}
        = \begin{pmatrix} J & \Azo \\ 0 & \Aoo \end{pmatrix}^T
    \end{equation*}
    because $\Aoo$ is symmetric.
\end{observation}

The decomposition of the order-demand matrix in the above Observation highlights the symmetry between flow-demand and order-demand matrix for the same \LOG:
In each case, one row of the block form is taken directly from the extended adjacency matrix, while the other row has an identity-like block on the diagonal and an all-zero block on the off-diagonal.
The big difference is that some of the rows in the order-demand matrix are then set to zero via multiplication by the diagonal matrix $P$.

By comparing with the extended adjacency matrix $\A$ in \eqref{eq:extended-adjacency} and recalling that $\A$ is symmetric by construction, we note furthermore that the matrices $M$ and $\fullN$ encode almost all the information of the extended adjacency matrix: the only parts missing are edges among inputs and edges among outputs (which can straightforwardly be seen to be irrelevant to the existence of flow).

\begin{theorem}
    Let $\Gamma=(G,I,O,\ld)$ be a labelled open graph which satisfies $\abs{I}=\abs{O}$.
    Let $\Gamma'$ be the reverse of $\Gamma$.
    Then $\Gamma$ has Pauli flow if and only if $\Gamma'$ has Pauli flow.
\end{theorem}
\begin{proof}
    Without loss of generality, assume $I \cap O = \emptyset$.
    
    $(\Rightarrow)$: First, suppose $\Gamma$ has Pauli flow.
    By Theorem~\ref{algebraic interpretation section start}, the inverse $C$ of its flow demand matrix $M$ exists and the product $NC$ with the order-demand matrix $N$ forms a DAG.
    
    Now by Observation~\ref{obs:flow-demand-block}, the flow demand matrix for $\Gamma$ is
    \begin{equation*}
        M = \begin{pNiceMatrix}[first-row, last-col]
           \Xlike\cup O & \Zlike & \\
           \Azz        & \Azo   & I\cup\Xlike \\
           0           & \I     & \Zlike
          \end{pNiceMatrix}
    \end{equation*}
    Break its inverse $C = M^{-1}$ into analogous blocks, i.e.\ let
    \begin{equation*}
        C = \begin{pNiceMatrix}[first-row, last-col]
            I\cup\Xlike & \Zlike & \\
            C_{00}      & C_{01} & \Xlike\cup O \\
            C_{10}      & C_{11} & \Zlike
        \end{pNiceMatrix}
    \end{equation*}
    such that
    \begin{equation*}
        \I = MC
        = \begin{pNiceMatrix}[first-row, last-col]
             I\cup\Xlike & \Zlike & \\
             \Azz C_{00} + \Azo C_{10} & \Azz C_{01} + \Azo C_{11} & I\cup\Xlike \\
             C_{10} & C_{11} & \Zlike
        \end{pNiceMatrix}
    \end{equation*}
    Then we have $C_{10} = 0$ and $C_{11}=\I$ by inspection.
    Thus the top left block implies $\I =\Azz\Czz$, i.e.\ $C_{00} = (\Azz)^{-1}$.
    The top right block becomes $0 = \Azz C_{01} + \Azo$.
    Hence $\Azz C_{01} = \Azo$ (as we are working over $\mathbb{F}_2$) and finally $C_{01} = (\Azz)^{-1} \Azo = \Czz \Azo$.
    To summarise,
    \begin{equation*}
        C = \begin{pNiceMatrix}[first-row, last-col]
            I\cup\Xlike & \Zlike   & \\
            \Czz        & \Czz\Azo & \Xlike\cup O \\
            0           & \I       & \Zlike
        \end{pNiceMatrix}
    \end{equation*}
    Again by Observation~\ref{obs:flow-demand-block}, the flow demand matrix for $\Gamma'$ is
    \begin{equation*}
        M'
        = \begin{pmatrix} \Azz & 0 \\ \Aoz & \I \end{pmatrix}^T
    \end{equation*}
    As $\Azz$ is invertible (we found $\Czz = (\Azz)^{-1}$) and because of its block structure, $M'$ is also invertible.
    By an argument analogous to the above, we find that
    \begin{equation*}
        C' = (M')^{-1} = \begin{pmatrix} \Czz & 0 \\ \Aoz\Czz & \I \end{pmatrix}^T
    \end{equation*}
    with $\Czz$ the same matrix as above.
    
    It remains to show that $N'C'$ is a DAG, knowing that $NC$ is a DAG.
    First, by Observation~\ref{obs:order-demand-block}:
    \begin{equation*}
        NC = P\fullN C
        = P \begin{pmatrix} J & 0 \\ \Aoz & \Aoo \end{pmatrix} \begin{pmatrix} \Czz & \Czz\Azo \\ 0 & \I \end{pmatrix}
        = P \begin{pmatrix}
            J\Czz & J\Czz\Azo \\
            \Aoz\Czz & \Aoz\Czz\Azo + \Aoo
        \end{pmatrix}
    \end{equation*}
    and similarly,
    \begin{align*}
        N'C'
        &= P'\fullN'C'
        = P' \begin{pmatrix} J & \Azo \\ 0 & \Aoo \end{pmatrix}^T \begin{pmatrix} \Czz & 0 \\ \Aoz\Czz & \I \end{pmatrix}^T
        = P' \left(\begin{pmatrix} \Czz & 0 \\ \Aoz\Czz & \I \end{pmatrix}\begin{pmatrix} J & \Azo \\ 0 & \Aoo \end{pmatrix} \right)^T \\
        &=P' \begin{pmatrix} \Czz J & \Czz\Azo \\ \Aoz\Czz J & \Aoz\Czz\Azo + \Aoo \end{pmatrix}^T
    \end{align*}
    Note that $(NC)_{vw} = 0$ whenever $v\in I\cup\paulis$: inputs and Pauli measured vertices are always initial in the partial order.
    This implies that the values in the \emph{rows} labelled by $I\cup\paulis$ are irrelevant to determining whether $NC$ is a DAG: it suffices to consider only the submatrix of $NC$ whose rows and columns are both labelled by $\planar$.
    The same holds for $N'C'$.
    
    Now, the three matrices $P$, $P'$, and $J$ are simply the identity matrix when restricted to rows and columns that are labelled as $\planar$.
    Thus the $\planar$ times $\planar$ submatrices of the two products satisfy:
    \begin{equation*}
        (N'C')_{\planar}^{\planar}
        = \left( \begin{pmatrix} \Czz J & \Czz\Azo \\ \Aoz\Czz J & \Aoz\Czz\Azo + \Aoo \end{pmatrix}^T \right)_{\planar}^{\planar}
        = \left( \begin{pmatrix} \Czz & \Czz\Azo \\ \Aoz\Czz & \Aoz\Czz\Azo + \Aoo \end{pmatrix}_{\planar}^{\planar} \right)^T
        = \left( (NC)_{\planar}^{\planar} \right)^T
    \end{equation*}
    since taking submatrices commutes with taking the transpose.
    By assumption of flow existence for $\Gamma$, $NC$ is a DAG.
    Hence $N'C'$ is a DAG and by Theorem~\ref{algebraic interpretation section start}, $\Gamma'$ has Pauli flow, ending the proof.
    
    $(\Leftarrow)$: Note that different measurement labels of the inputs do not change the flow-demand and order-demand matrices of $\Gamma$, and hence by Theorem~\ref{algebraic interpretation section start} they do not change the existence of flow.
    Thus without loss of generality, we can assume that $\lambda(v) = XY$ for all $v \in I$.
    Using the previous part of the proof, if $\Gamma'$ has Pauli flow, then so does $\Gamma''$.
    However, $\Gamma'' = \Gamma$ (up to a potential and irrelevant difference in measurement labels of inputs), so $\Gamma$ has Pauli flow, ending the proof.
\end{proof}

\begin{corollary}
    Let $\Gamma=(G,I,O,\ld)$ be a labelled open graph which satisfies $\abs{I}=\abs{O}$.
    Suppose $\Gamma$ has focused Pauli flow $(c,\prec_c)$ and its reverse $\Gamma'$ has focused Pauli flow $(c',\prec_{c'})$, then:
    \begin{itemize}
        \item For all $u\in I\cup\Xlike$ and $w\in \Xlike\cup O$, we have $w\in c(u) \Leftrightarrow u\in c'(w)$.
        \item For all $u,w\in\planar$, we have $u\prec_c w \Leftrightarrow w\prec_{c'} u$.
    \end{itemize}
\end{corollary}

The complexity of finding a reverse flow requires matrix multiplication, so it is not fundamentally easier than finding a flow from scratch.
Yet if a labelled open graph with equal numbers of inputs and outputs is dominated by vertices in $\Xlike$, then it will be more efficient to find the flow of the reverse graph by taking the matrix $\Czz$ (which describes the flow on the $X$-like part of the original labelled open graph) and multiplying it with the (smaller) submatrices of the extended adjacency matrix.

When the number of inputs is smaller than the number of outputs, the flow cannot be reversed -- the flow-demand matrix of the reverse labelled open graph must have more rows than columns and thus cannot be right-reversible.
For example, the flow from Figure~\ref{fig:pf} cannot be reversed.
Similarly, the flow in such a case does not need to be unique as the flow-demand matrix may have many different right inverses.
In the future algorithm, we parametrise all such right inverses by relating them via focussed sets which study next.

\subsection{Focused sets}

The flow-demand matrix has another connection to Pauli flow that we explore in this subsection, this time not via its inverse but via its kernel. First, we define focused sets: sets with trivial net effect when used for corrections.

\begin{definition}[Focused set {\cite[part of Definition 4.3]{simmonsRelatingMeasurementPatterns2021}}]\label{def:focused set} 
    Given the labelled open graph $\Gamma = (G,I,O,\lambda)$, a set $\someset \subseteq \comp{I}$ is \emph{focused over} $S \subseteq{\comp{O}}$ if:\begin{enumerate}
        \item[{\crtcrossreflabel{(Fs1)}[Fs1]}] $\forall w \in S \cap \someset . \lambda(w) \in \{ XY, X, Y \}$
        \item[{\crtcrossreflabel{(Fs2)}[Fs2]}] $\forall w \in S \cap \odd{\someset} . \lambda(w) \in \{ XZ, YZ, Y, Z \}$
        \item[{\crtcrossreflabel{(Fs3)}[Fs3]}] $\forall w \in S . \lambda(w) = Y \Rightarrow w \notin \codd{\someset}$, note $w \notin \codd{\someset}$ is equivalent to $w \in \someset \Leftrightarrow w \in \odd{\someset}$.
    \end{enumerate}
    A set is simply called \emph{focused} if it is focused over $\comp{O}$. The set of all focused sets for $\Gamma$ is denoted $\mathfrak{F}_{\Gamma}$.
\end{definition}

The focusing conditions for the Pauli flow (cf.\ Definition~\ref{def:focused}) equivalently say that the correction set $c(v)$ is focused over $\comp{O}\setminus\{v\}$ for all $v \in \comp{O}$. The above definition was introduced by Simmons, who also proved that the focused sets for a given labelled open graph form a group \cite[Lemma B.8]{simmonsRelatingMeasurementPatterns2021}. Here, we provide an alternative interpretation for the collection of focused sets, by providing an isomorphism with the kernel of the flow-demand matrix, so that the focused sets form a vector space. First, we show that given some set $\someset\sse\comp{I}$ one can find a maximal set over which $\someset$ is focused by considering the product $M\ind{\someset}{\comp{I}}$.

\begin{lemma}\label{M product gives focusing condition}
    Let $\Gamma = (G,I,O,\lambda)$ be a labelled open graph and $\someset \subseteq \comp{I}$ be any subset of non-inputs. Then, $\someset$ is focused over $\supp{\ind{\comp{O}}{\comp{O}} + M\ind{\someset}{\comp{I}}}$ and $\someset$ is not focused over any larger subset of the non-outputs.
\end{lemma}

\begin{proof}
    Before beginning the proof, we will explain the notation $\supp{\ind{\comp{O}}{\comp{O}} + M\ind{\someset}{\comp{I}}}$, which can be unpacked as follows: $M\ind{\someset}{\comp{I}}$ is the result of multiplying the flow-demand matrix by the vector corresponding to the set $\someset$. This vector is in the space $\mathbb{F}_2^{\comp{O}}$. Next, we add to it $\ind{\comp{O}}{\comp{O}}$, i.e.\ the vector of all-$1$s in $\mathbb{F}_2^{\comp{O}}$. Hence, we are flipping all values in $M\ind{\someset}{\comp{I}}$. Next, we take the support of this vector: these are the vertices $v$ such that the flipped vector has a $1$ for its $v$-th entry. This means, $M\ind{\someset}{\comp{I}}$ has a $0$ for its $v$-th entry. Thus, the lemma claims that $\someset$ is focused over the set of vertices corresponding to the $0$-valued entries of $M\ind{\someset}{\comp{I}}$.
    Considering a vector as a linear map via the inner product, these entries correspond to the \textit{kernel of} $M\ind{\someset}{\comp{I}}$.
    We are now ready to proceed to the proof itself.

    Let $F_{\someset} = \supp{\ind{\comp{O}}{\comp{O}} + M\ind{\someset}{\comp{I}}}$.
    By inspecting conditions \ref{Fs1}-\ref{Fs3}, for any sets $S,S'\sse\comp{O}$, $\someset$ is simultaneously focused over $S$ and focused over $S'$ if and only if it is focused over $S \cup S'$.
    Therefore, it is sufficient to show that for all $v \in \comp{O}$: $v \in F_{\someset}$ if and only if $\someset$ is focused over $\{v\}$.

    $(\Rightarrow)$: Let $v \in F_{\someset}$. Then $1 = \left(\ind{\comp{O}}{\comp{O}} + M\ind{\someset}{\comp{I}}\right)_v  = 1 + \left(M\ind{\someset}{\comp{I}}\right)_v$. Hence, $0 = \left(M\ind{\someset}{\comp{I}}\right)_v$. Therefore $M_{v,*} \ind{\someset}{\comp{I}} = 0$ and by Lemma~\ref{row by col meaning}, we get:\begin{align*}
        \triplecase
            {\lambda(v) \in \{X, XY\}}{v \notin \odd{\someset}}
            {\lambda(v) \in \{XZ, YZ, Z \}}{v \notin \someset}
            {\lambda(v)=Y}{v \notin \codd{\someset}}
    \end{align*}
    Hence $\someset$ satisfies all three conditions \ref{Fs1}, \ref{Fs2}, and \ref{Fs3} over the set $\{v\}$, as required.

    $(\Leftarrow)$: Let $\someset$ be focused over $\{v\}$. Then, by \ref{Fs2}:\begin{align*}
        v \notin \odd{\someset}\quad& \text{when}\quad \lambda(v) \in \{ X, XY \}.
    \intertext{Similarly, by \ref{Fs1}:}
        v \notin \someset\quad& \text{when}\quad \lambda(v) \in \{ XZ, YZ, Z \},
    \intertext{and by \ref{Fs3}:}
        v \notin \codd{\someset}\quad& \text{when}\quad \lambda(v) = Y.
    \end{align*}
    Hence, again using Lemma~\ref{row by col meaning}, we get $M_{v,*} \ind{\someset}{\comp{I}} = 0$ and thus $v \in F_{\someset}$.
\end{proof}

Using the above lemma, we find an isomorphism between the kernel of the flow-demand matrix and the collection of focused sets.

\begin{theorem}\label{focused sets are kernel of M}
    Let $\Gamma = (G,I,O,\lambda)$ be a labelled open graph and let $M$ be its flow-demand matrix. Then $\mathfrak{F}_{\Gamma} \cong \ker M_{\Gamma}$.
\end{theorem}

\begin{proof}
    Suppose $\someset \in \mathfrak{F}$.
    By Lemma~\ref{M product gives focusing condition}, any set $\someset$ is focused over $\supp{\ind{\comp{O}}{\comp{O}} + M\ind{\someset}{\comp{I}}}$. Combining these, we find $\supp{\ind{\comp{O}}{\comp{O}} + M\ind{\someset}{\comp{I}}} = \comp{O}$, which implies $\left(M\ind{\someset}{\comp{I}}\right)_v = 0$ for all $v \in \comp{O}$.
    In other words, $\ind{\someset}{\comp{I}} \in \ker M$.

    Now suppose $w\in\ker M$.
    Again by Lemma~\ref{M product gives focusing condition}, this implies that $\supp{w}$ is focused over $\comp{O}$, i.e.\ $\supp{w}$ is a focused set.

    Therefore, an isomorphism between $\mathfrak{F}$ and $\ker M$ is given by the indicator function $\ind{\bullet}{\comp{I}}$ and its inverse is given by the support $\supp{\bullet}$.
\end{proof}

For example, the kernel of the flow-demand matrix in Figure~\ref{fig:M} consists of two vectors:\begin{equation*}
    \begin{pmatrix}
        0 & 0 & 0 & 0 & 0 & 0
    \end{pmatrix}^T\text{ and }\begin{pmatrix}
        0 & 0 & 1 & 0 & 1 & 1
    \end{pmatrix}^T
\end{equation*}
corresponding to focused sets $\emptyset$ and $\{ e,o_1,o_2 \}$, respectively.

\section{Algorithms}\label{sec:algorithms}

We provide an important application for the newly defined algebraic interpretation -- we improve the existing flow-finding algorithms. First, we consider a special case where the numbers of inputs and outputs are equal. Afterwards, we consider the general case. While it would be sufficient in principle to consider the general case only, it is nevertheless worth distinguishing the two cases, due to the remarkably simple algorithm in the case of $n_I = n_O$. For the latter case, we also provide a lower bound on the computational complexity of the problem of finding Pauli flow.

Firstly, we reformulate the Theorem~\ref{algebraic interpretation section start} to remove the assumption of a known correction function, in order to provide the groundwork for the flow-finding algorithms.

\begin{corollary}\label{algebraic interpretation corollary}
    Let $\Gamma = (G,I,O,\lambda)$ be a labelled open graph.
    Let $M$ be the flow-demand matrix of $\Gamma$ and let $N$ be the order-demand matrix of $\Gamma$.
    Then $\Gamma$ has Pauli flow if and only if there exists a correction matrix $C$ such that:\begin{itemize}
        \item $C$ has shape $(n-n_I) \times (n-n_O)$ with rows corresponding to $\comp{I}$ and columns corresponding to $\comp{O}$,
        \item $MC = Id_{\comp{O}}$, and
        \item $NC$ is the adjacency matrix of some directed acyclic graph.
    \end{itemize}
    Furthermore, for any such matrix $C$, the corresponding minimal focused Pauli flow $(c,\prec_c)$ is given as follows, where $v,w \in \comp{O}$:
    \begin{itemize}
     \item $c(v) := \{ u \in \comp{I} \mid C_{u,v} = 1 \}$,
     \item $v \trl_c w \Leftrightarrow (NC)_{w,v} = 1$, and
     \item $\prec_c$ is the transitive closure of $\trl_c$.
    \end{itemize}

\end{corollary}
\begin{proof}
    By inspection, follows from the theorems in Section~\ref{sec:new algebraic interpretation}.
\end{proof}

\subsection{Equal number of inputs and outputs}

Suppose, that we are given labelled open graph $\Gamma = (G,I,O,\lambda)$ with $n_I = n_O$, i.e.\ $|I| = |O|$. Then, the flow-demand matrix is square, and Corollary~\ref{algebraic interpretation corollary} gives rise to a simple and efficient algorithm for Pauli flow detection.

\begin{theorem}\label{algo easy}
    Let $\Gamma = (G,I,O,\lambda)$ be a labelled open graph with $n_I = n_O$. Then, there exists an $\bigO(n^3)$ algorithm which either finds a Pauli flow or determines that no Pauli flow exists.
\end{theorem}

\begin{proof}
    Proceed as follows:\begin{enumerate}
        \item Construct the flow-demand matrix $M$ according to Definition~\ref{flow-demand matrix definition}.
        \item Construct the order-demand matrix $N$ according to Definition~\ref{order-demand matrix definition}.
        \item Check if $M$ is invertible.\begin{itemize}
            \item If not, return that there is no flow.
        \end{itemize}
        \item Otherwise, compute the unique inverse $C$ of $M$.
        \item Compute the matrix product $NC$.
        \item Check if $NC$ is a DAG.\begin{itemize}
            \item If not, return that there is no flow.
        \end{itemize}
        \item Otherwise, return the tuple $(C, NC)$: these form the matrix encodings of the focused extensive correction function $c$ and the corresponding induced relation $\trl_c$.
    \end{enumerate}

    The above procedure correctly finds flow: by Corollary~\ref{algebraic interpretation corollary}, any focused extensive correction matrix $C$ must be a right inverse of $M$. However, since $M$ is square, any right inverse is a two-sided inverse and is unique. Then, the only other condition that must be checked is whether $NC$ forms a DAG, as in the procedure above.

    To see that the complexity of the above procedure is $\bigO(n^3)$, consider:\begin{itemize}
        \item Steps 1--2 can be implemented in $\bigO(n^2)$, i.e.\ in time linear in the size of the constructed matrices.
        \item Steps 3--4 can be achieved by performing Gaussian elimination, which runs in $\bigO(n^3)$ time.
        \item Step 5 requires matrix multiplication, with a na\"ive implementation again taking $\bigO(n^3)$ time.
        \item Checking if a graph is a DAG can be done in $\bigO(n^2)$ by running depth-first search algorithm or similar (see \cite[Section 20.4]{cormenIntroductionAlgorithmsFourth2022}).
    \end{itemize}

    Thus, the total runtime is dominated by Gaussian elimination and matrix multiplication running in $\bigO(n^3)$, with the remaining steps running in $\bigO(n^2)$. If one requires the full partial order $\prec_c$ rather than only $\trl_c$, then the transitive closure of $\trl_c$ must be computed. This again can be done in $\bigO(n^2)$ time, so it does not affect the overall complexity.
\end{proof}

The corresponding pseudocode can be found in Appendix~\ref{pseudocodes easy}. The above theorem already gives us a faster method for finding extended gflow and Pauli flow than previously known procedures: $\bigO(n^4)$ from \cite{backensThereBackAgain2021} for extended gflow and $\bigO(n^5)$ from \cite{simmonsRelatingMeasurementPatterns2021} for Pauli flow.
Yet it may still be possible to speed our algorithm up further.

Faster methods for matrix inversion and multiplication exist, for instance by utilization of Strassen's algorithm \cite{strassenGaussianEliminationNot1969}. Using such methods reduces the complexity to $\bigO\left( n^{\log_2 7} \right) \approx \bigO\left( 2^{2.807} \right)$. Due to the constant factors involved, the method becomes beneficial only for sufficiently large matrices, i.e.\ $n$ must be on the order of hundreds or thousands to notice any difference \cite{huangStrassenAlgorithmReloaded2016}. Even faster algorithms for matrix multiplication and inversion exist (for instance \cite{williamsNewBoundsMatrix2024}), but they are \textit{galactic}, meaning impractical due to their requirements on the size of $n$.

Moreover, we assumed nothing about the structure of the simple graph underlying $\Gamma$. Depending on its structure (and thus the properties of the matrices derived from it), faster methods may exist. For instance, when the underlying graph is sparse, so is the flow-demand matrix. Finding the inverse of a sparse matrix can be solved even faster than fast matrix multiplication, for example, see \cite{casacubertaFasterSparseMatrix2022}.

The following lower bound on the problem of flow detection allows us to argue that, in general, further improvements to the detection of the Pauli flow problem when $n_I = n_O$ require better methods for finding the inverse of a matrix over a field $\mathbb{F}_2$.

\begin{theorem}\label{lower bound}
    The problem of finding Pauli flow is at least as computationally expensive as the problem of finding the inverse of a matrix over $\mathbb{F}_2$.
\end{theorem}

\begin{proof}
    Let $M$ be an arbitrary $n \times n$ matrix over $\mathbb{F}_2$. Let $\Gamma = (G,I,O,\lambda)$ be defined as follows:\begin{align*}
        I &= \{ i_1, \dots, i_n \}\\
        O &= \{ o_1, \dots, o_n \}\\
        V &= I \cup O\\
        E &= \{ (i_k, o_\ell) \mid M_{k,\ell} = 1 \text{ for } k,\ell \in \{1, \dots, n\} \}\\
        G &= (V,E)\\
        \lambda(i_k) &= X \quad (\forall k \in \{1, \dots, n\})
    \end{align*}
    Then, the flow-demand matrix $M_{\Gamma}$ equals $M$ and the order-demand matrix $N_{\Gamma}$ is identically $0$. Hence, by Theorem~\ref{algebraic interpretation corollary}, $\Gamma$ has Pauli flow if and only if $M$ is invertible and the correction function is encoded by the inverse of $M$.
\end{proof}

For the moment, no better lower bound for finding the inverse of a matrix than $\Omega(n^2)$ is known\footnote{Over real or complex numbers, there exists a bound $\Omega\left(n^2 \log_2 n\right)$, but it does not apply to finite fields (see \cite[Section 28.2]{cormenIntroductionAlgorithmsFourth2022})}. Therefore, finding Pauli flow must take at least $\Omega(n^2)$ operations. In particular, it can be no better than the best-known method for finding causal flow from \cite{mhallaFindingOptimalFlows2008a} with a runtime of $\bigO(m)$ where $m$ is the number of edges.

\subsection{Different number of inputs and outputs}

If $n_O < n_I$, then the flow-demand matrix has more rows than columns and hence it never has a right inverse, so there is no flow. Thus, we only need to consider the case $n_I < n_O$, where the situation is more complex. The lower bound proved in Theorem~\ref{lower bound} applies. We also still obtain a $\bigO(n^3)$ algorithm for flow detection and finding. However, the algorithm is no longer as simple as computing the inverse and verifying whether one matrix product is a DAG.

While the focused correction matrix must be a right inverse of the flow-demand matrix, this right inverse is not unique.
To identify a right inverse that yields a DAG when multiplied by the order-demand matrix, or conclude that no such right inverse exists, we modify the flow-finding algorithm from the previous section in two ways.
Firstly, we perform a change of basis that gives a more convenient parametrisation of the right inverses.
Secondly, we use an approach similar to the one from \cite{mhallaFindingOptimalFlows2008a}, which was subsequently extended in \cite{backensThereBackAgain2021, simmonsRelatingMeasurementPatterns2021} -- we define the correction function in layers of the partial order, starting with the last layer and moving towards the first vertices in the order. In the first step, we define correction sets for all vertices that are final in the order, i.e.\ vertices $v$ such that for all vertices $u$ we have $\neg v \prec u$. Next, we find correction sets for vertices that can only be succeeded by vertices from the first layer et cetera. If, at any point, we cannot correct any vertices, there is no flow. The methods used in \cite{mhallaFindingOptimalFlows2008a,backensThereBackAgain2021} required Gaussian elimination for each layer of vertices, while in \cite{simmonsRelatingMeasurementPatterns2021} up to $\bigO(n)$ Gaussian eliminations per layer were necessary. Since the number of layers could be $\bigO(n)$, it leads to upper bounds of, at best, $\bigO(n^4)$. We improve on this approach by reducing number of operations performed at each layer to $\bigO(n^2 s)$ where $s$ is the number of vertices in the layer. This leads to a total cost of $\bigO(n^3)$ operations.

The following theorem and its proof describe the Pauli flow-finding algorithm in prose and explain its functioning as well as its complexity.
A pseudocode presentation of the same algorithm may be found in Appendix~\ref{pseudocodes hard}.

\begin{theorem}\label{algo hard}
    Let $\Gamma = (G,I,O,\lambda)$ be a labelled open graph with $n_I \le n_O$, i.e.\ $|I|\le|O|$. Then, there exists an $\bigO(n^3)$ algorithm which either finds a Pauli flow or determines that no Pauli flow exists.
\end{theorem}

\begin{proof}
    The proof consists of providing an algorithm, showing its correctness, and verifying the complexity bound. We interleave the three parts: for each step of the algorithm, we first provide the procedure and then explain it and analyse its complexity.
    \begin{enumerate}
        \item\label{step:M}
            \proc{Construct the flow-demand matrix $M$ according to Definition~\ref{flow-demand matrix definition}. Recall this matrix has size $(n - n_O) \times (n - n_I)$.}
        
        \item\label{step:N}
            \proc{Construct the order-demand matrix $N$ according to Definition~\ref{order-demand matrix definition}.}

        \item\label{step:M right-invertibility}
            \proc{Check whether $M$ is right-invertible, i.e.\ whether $\rank M = n-n_O$. \begin{itemize}
                \item If $M$ is not right-invertible, return that there is no flow.
            \end{itemize}}
            
            \expl{By Corollary~\ref{algebraic interpretation corollary}, flow existence necessarily means that $M$ is right-invertible. $M$ is right-invertible if all of its rows are linearly independent, i.e.\ if $\rank M$ equals the number of rows of $M$, that is $n-n_O$. The check for right invertibility can be achieved by performing Gaussian elimination, which runs in time $\bigO(n^3)$.}

        \item\label{step:find-inverse}
        \proc{Find any right inverse of $M$ and call it $C_0$.}

        \expl{By the previous step, a right inverse exists. It can be found for instance by Gaussian elimination and backtracking, again in $\bigO(n^3)$ and has size $(n - n_I) \times (n - n_O)$.

        Even if there is a Pauli flow, the matrix $C_0$ will not necessarily be such that $NC_0$ is a DAG.
        Yet by combining it with information about the kernel of $M$, we will be able to parametrise all the right inverses of $M$:
        Any right inverse of $M$ must necessarily be of the form $C_0 + F'$ where $F'$ is a matrix whose columns are vectors in $\ker M$. Such matrices work as then $M(C_0+F') = Id_{\comp{O}} + \mathbf{0} = Id_{\comp{O}}$. They are also necessary, as if $MC = Id_{\comp{O}}$, then $F' = C - C_0$ must satisfy both $C = C_0 + F'$ and $\mathbf{0} = MC - MC_0 = M(C - C_0) = MF'$.

        In the next three steps, we will first find the kernel of $M$ and then use it to perform a change of basis, that gives a simpler form to the possible right inverses.
        This simplifies the search for a right inverse that satisfies all the Pauli flow conditions.}

        \item\label{step:find-kernel}
        \proc{Find a matrix $F$ whose columns form a basis of $\ker M$.}

        \expl{This can again be found with Gaussian elimination.
        By the rank-nullity theorem, the dimension of $\ker M$ is $n_O-n_I$, hence $F$ has size $(n - n_I) \times (n_O - n_I)$.}
        
        \item\label{step:c'} \proc{Let $C' = \left[ C_0 \mid F \right]$.}

        \expl{Note that $C'$ is a square matrix of size $(n-n_I) \times (n-n_I)$.
        We know $MC_0 = Id$, so the columns of $C_0$ form a basis of $\supp{M}$. Similarly, the columns of $F$ form a basis of $\ker M$. Thus, the union of the columns of both forms a basis of $\mathbb{F}^{n-n_I}$, i.e.\ the columns of $C'$ are linearly independent. Therefore, $\rank C' = n-n_I$ and $C'$ is invertible.
        From now on, we will use diagrams to illustrate the sizes of the matrices we work with.
        When the rows or columns correspond to sets of vertices, we also indicate them on the diagrams.
        \begin{align*}
            \begin{tikzpicture}
	\begin{pgfonlayer}{nodelayer}
		\node [style=none] (0) at (-2, 1.25) {};
		\node [style=none] (1) at (2, 1.25) {};
		\node [style=none] (2) at (-2, -1.25) {};
		\node [style=none] (3) at (2, -1.25) {};
		\node [style=none] (4) at (0, 0) {$M$};
		\node [style=none] (6) at (2.5, 1.25) {};
		\node [style=none] (7) at (2.5, -1.25) {};
		\node [style=none] (8) at (-2, 1.75) {};
		\node [style=none] (9) at (2, 1.75) {};
		\node [style=none] (10) at (0, 2.25) {$\bar{I}$};
		\node [style=none] (11) at (3, 0) {$\bar{O}$};
		\node [style=none] (12) at (-6, 0) {};
		\node [style=none] (13) at (6, 0) {};
		\node [style=none] (14) at (0, 4) {};
		\node [style=none] (15) at (0, -3) {};
	\end{pgfonlayer}
	\begin{pgfonlayer}{edgelayer}
		\draw (1.center) to (0.center);
		\draw (2.center) to (3.center);
		\draw [style=brace edge] (6.center) to (7.center);
		\draw [style=brace edge] (8.center) to (9.center);
		\draw (1.center) to (3.center);
		\draw (2.center) to (0.center);
	\end{pgfonlayer}
\end{tikzpicture}\quad\begin{tikzpicture}
	\begin{pgfonlayer}{nodelayer}
		\node [style=none] (0) at (-2, -2) {};
		\node [style=none] (1) at (-2, 2) {};
		\node [style=none] (2) at (2, -2) {};
		\node [style=none] (3) at (2, 2) {};
		\node [style=none] (4) at (0.5, -2) {};
		\node [style=none] (5) at (0.5, 2) {};
		\node [style=none] (6) at (1.25, 0) {$F$};
		\node [style=none] (7) at (-0.75, 0) {$C_0$};
		\node [style=none] (8) at (-3, 0) {$C'=$};
		\node [style=none] (9) at (-2, 2.5) {};
		\node [style=none] (10) at (0.5, 2.5) {};
		\node [style=none] (11) at (2, 2.5) {};
		\node [style=none] (12) at (2.5, 2) {};
		\node [style=none] (13) at (2.5, -2) {};
		\node [style=none] (14) at (3, 0) {$\bar{I}$};
		\node [style=none] (15) at (-0.75, 3) {$\bar{O}$};
		\node [style=none] (16) at (1.25, 3) {$n_O-n_I$};
		\node [style=none] (17) at (-6, 0) {};
		\node [style=none] (18) at (6, 0) {};
		\node [style=none] (19) at (0, 4) {};
		\node [style=none] (20) at (0, -3) {};
	\end{pgfonlayer}
	\begin{pgfonlayer}{edgelayer}
		\draw (0.center) to (2.center);
		\draw (2.center) to (3.center);
		\draw (3.center) to (1.center);
		\draw (1.center) to (0.center);
		\draw (4.center) to (5.center);
		\draw [style=brace edge] (9.center) to (10.center);
		\draw [style=brace edge] (10.center) to (11.center);
		\draw [style=brace edge] (12.center) to (13.center);
	\end{pgfonlayer}
\end{tikzpicture}
        \end{align*}
        }

        \item\label{step:basis-change} \proc{Compute $\rowbasechange{N}{\mathcal{B}} = NC'$.}

        \expl{As explained in Step~\ref{step:find-inverse}, the right inverse $C_0$ is not necessarily unique. However, we know the form of all possible right inverses, which can be parametrised in terms of $C_0$ and the columns of $F$.

        The problem is to find some right inverse $C$ such that $NC$ is a DAG. Brute force checking all possible right inverses cannot be performed, as the number of right inverses grows exponentially in $n_O - n_I$. Instead, we simplify the form of right inverses by the basis change just performed.

        To see this, note that the equivalent basis change on $M$ yields $\rowbasechange{M}{\mathcal{B}} = MC' = \left[ MC_0 \mid MF \right] = \left[ Id_{\comp{O}} \mid \mathbf{0} \right]$, since $C_0$ is the right inverse of $M$ and $F$ is the basis of $\ker M$. This makes all right inverses of $\rowbasechange{M}{\mathcal{B}}$ very simple: their first $n - n_O$ rows form $Id_{\comp{O}}$ and the remaining rows can contain any values. In the following picture, the part of the matrix that can have any values is denoted $\starpart$.
        \begin{equation}\label{eq:image-c-b}
            \begin{tikzpicture}
	\begin{pgfonlayer}{nodelayer}
		\node [style=none] (0) at (-2, 1.25) {};
		\node [style=none] (1) at (2, 1.25) {};
		\node [style=none] (4) at (-2, -1.25) {};
		\node [style=none] (5) at (2, -1.25) {};
		\node [style=none] (7) at (-0.75, 0) {$Id_{\bar{O}}$};
		\node [style=none] (8) at (-3, 0) {$M_{\mathcal{B}}=$};
		\node [style=none] (9) at (2.5, 1.25) {};
		\node [style=none] (10) at (2.5, -1.25) {};
		\node [style=none] (12) at (-2, 1.75) {};
		\node [style=none] (13) at (0.5, 1.75) {};
		\node [style=none] (14) at (-0.75, 2.25) {$\bar{O}$};
		\node [style=none] (15) at (3, 0) {$\bar{O}$};
		\node [style=none] (17) at (-6, 0) {};
		\node [style=none] (18) at (6, 0) {};
		\node [style=none] (19) at (0, 4) {};
		\node [style=none] (20) at (0, -3) {};
		\node [style=none] (21) at (0.5, 1.25) {};
		\node [style=none] (22) at (0.5, -1.25) {};
		\node [style=none] (23) at (1.25, 0) {$\mathbf{0}$};
		\node [style=none] (24) at (2, 1.75) {};
		\node [style=none] (25) at (1.25, 2.25) {$n_O-n_I$};
	\end{pgfonlayer}
	\begin{pgfonlayer}{edgelayer}
		\draw (1.center) to (0.center);
		\draw (4.center) to (5.center);
		\draw [style=brace edge] (9.center) to (10.center);
		\draw [style=brace edge] (12.center) to (13.center);
		\draw (0.center) to (4.center);
		\draw (1.center) to (5.center);
		\draw (22.center) to (21.center);
		\draw [style=brace edge] (13.center) to (24.center);
	\end{pgfonlayer}
\end{tikzpicture}\quad\begin{tikzpicture}
	\begin{pgfonlayer}{nodelayer}
		\node [style=none] (0) at (-2, -2) {};
		\node [style=none] (1) at (-2, 2) {};
		\node [style=none] (4) at (0.5, -2) {};
		\node [style=none] (5) at (0.5, 2) {};
		\node [style=none] (7) at (-0.75, 0.75) {$Id_{\bar{O}}$};
		\node [style=none] (8) at (-3, 0) {$C^{\mathcal{B}}=$};
		\node [style=none] (9) at (-2, 2.5) {};
		\node [style=none] (10) at (0.5, 2.5) {};
		\node [style=none] (15) at (-0.75, 3) {$\bar{O}$};
		\node [style=none] (17) at (-6, 0) {};
		\node [style=none] (18) at (6, 0) {};
		\node [style=none] (19) at (0, 4) {};
		\node [style=none] (20) at (0, -3) {};
		\node [style=none] (21) at (0.5, -0.5) {};
		\node [style=none] (22) at (-2, -0.5) {};
		\node [style=none] (23) at (1, 2) {};
		\node [style=none] (24) at (1, -0.5) {};
		\node [style=none] (25) at (1.5, 0.75) {$\bar{O}$};
		\node [style=none] (26) at (1, -2) {};
		\node [style=none] (27) at (2.25, -1.25) {$n_O-n_I$};
		\node [style=none] (28) at (-0.75, -1.25) {$\starpart$};
	\end{pgfonlayer}
	\begin{pgfonlayer}{edgelayer}
		\draw (1.center) to (0.center);
		\draw (4.center) to (5.center);
		\draw [style=brace edge] (9.center) to (10.center);
		\draw (5.center) to (1.center);
		\draw (0.center) to (4.center);
		\draw (22.center) to (21.center);
		\draw [style=brace edge] (23.center) to (24.center);
		\draw [style=brace edge] (24.center) to (26.center);
	\end{pgfonlayer}
\end{tikzpicture}
        \end{equation}
        Now, suppose that we find $\colbasechange{C}{\mathcal{B}}$ such that $\rowbasechange{M}{\mathcal{B}} \colbasechange{C}{\mathcal{B}} = Id_{\comp{O}}$ and $\rowbasechange{N}{\mathcal{B}} \colbasechange{C}{\mathcal{B}}$ is a DAG. Then, we can find the desired correction matrix $C$ encoding a focused extensive correction function $c$ by computing $C'\colbasechange{C}{\mathcal{B}}$, since:\begin{gather*}
            \rowbasechange{M}{\mathcal{B}} \colbasechange{C}{\mathcal{B}} = (MC') \colbasechange{C}{\mathcal{B}} = M (C'\colbasechange{C}{\mathcal{B}})\\
            \rowbasechange{N}{\mathcal{B}} \colbasechange{C}{\mathcal{B}} = (NC') \colbasechange{C}{\mathcal{B}} = N (C'\colbasechange{C}{\mathcal{B}})
        \end{gather*}
        and so $M (C'\colbasechange{C}{\mathcal{B}}) = Id_{\comp{O}}$ and $N (C'\colbasechange{C}{\mathcal{B}})$ is a DAG.

        This approach also captures any possible solution as given a working $C$ one can find a working $\colbasechange{C}{\mathcal{B}}$ by computing $\left(C'\right)^{-1}C$. This is to say, that instead of solving the problem given $M$ and $N$, we can solve the problem given $\rowbasechange{M}{\mathcal{B}}$ and $\rowbasechange{N}{\mathcal{B}}$.}

        \item\label{step:NL-and-NR} \proc{Define $N_L$ and $N_R$ as the submatrices of $\rowbasechange{N}{\mathcal{B}}$ given by first $n-n_O$ columns and the remaining $n_O - n_I$ columns:
        \begin{align*}
            \begin{tikzpicture}
	\begin{pgfonlayer}{nodelayer}
		\node [style=none] (0) at (-2, 1.25) {};
		\node [style=none] (1) at (2, 1.25) {};
		\node [style=none] (2) at (-2, -1.25) {};
		\node [style=none] (3) at (2, -1.25) {};
		\node [style=none] (4) at (-0.75, 0) {$N_L$};
		\node [style=none] (5) at (-3, 0) {$N_{\mathcal{B}}=$};
		\node [style=none] (6) at (2.5, 1.25) {};
		\node [style=none] (7) at (2.5, -1.25) {};
		\node [style=none] (8) at (-2, 1.75) {};
		\node [style=none] (9) at (0.5, 1.75) {};
		\node [style=none] (10) at (-0.75, 2.25) {$\bar{O}$};
		\node [style=none] (11) at (3, 0) {$\bar{O}$};
		\node [style=none] (12) at (-6, 0) {};
		\node [style=none] (13) at (6, 0) {};
		\node [style=none] (14) at (0, 4) {};
		\node [style=none] (15) at (0, -3) {};
		\node [style=none] (16) at (0.5, 1.25) {};
		\node [style=none] (17) at (0.5, -1.25) {};
		\node [style=none] (18) at (1.25, 0) {$N_R$};
		\node [style=none] (19) at (2, 1.75) {};
		\node [style=none] (20) at (1.25, 2.25) {$n_O-n_I$};
	\end{pgfonlayer}
	\begin{pgfonlayer}{edgelayer}
		\draw (1.center) to (0.center);
		\draw (2.center) to (3.center);
		\draw [style=brace edge] (6.center) to (7.center);
		\draw [style=brace edge] (8.center) to (9.center);
		\draw (0.center) to (2.center);
		\draw (1.center) to (3.center);
		\draw (17.center) to (16.center);
		\draw [style=brace edge] (9.center) to (19.center);
	\end{pgfonlayer}
\end{tikzpicture}
        \end{align*}
        }
        \expl{From the condition $\rowbasechange{M}{\mathcal{B}} \colbasechange{C}{\mathcal{B}} = Id_{\comp{O}}$, we already know the first $n - n_O$ rows of $\colbasechange{C}{\mathcal{B}}$. This means, that we are able to express the product $\rowbasechange{N}{\mathcal{B}} \colbasechange{C}{\mathcal{B}}$ as follows:\begin{align*}
            \rowbasechange{N}{\mathcal{B}} \colbasechange{C}{\mathcal{B}} = N_L Id_{\comp{O}} + N_R \starpart = N_L + N_R \starpart
        \end{align*}
        Thus, the problem becomes to find a $(n_O-n_I)\times (n-n_O)$ matrix $\starpart$ such that $N_L + N_R \starpart$ is a DAG:
        \begin{align*}
            \begin{tikzpicture}
	\begin{pgfonlayer}{nodelayer}
		\node [style=none] (0) at (-4.5, 1.25) {};
		\node [style=none] (1) at (1.5, 1.25) {};
		\node [style=none] (2) at (-4.5, -1.25) {};
		\node [style=none] (3) at (1.5, -1.25) {};
		\node [style=none] (4) at (-3.25, 0) {$N_L$};
		\node [style=none] (12) at (-5.5, 0) {};
		\node [style=none] (13) at (5.5, 0) {};
		\node [style=none] (14) at (0, 2) {};
		\node [style=none] (15) at (0, -2) {};
		\node [style=none] (16) at (-2, 1.25) {};
		\node [style=none] (17) at (-2, -1.25) {};
		\node [style=none] (18) at (0.75, 0) {$N_R$};
		\node [style=none] (22) at (0, 1.25) {};
		\node [style=none] (23) at (0, -1.25) {};
		\node [style=none] (24) at (2, -0.75) {};
		\node [style=none] (25) at (4.5, -0.75) {};
		\node [style=none] (26) at (4.5, 0.75) {};
		\node [style=none] (27) at (2, 0.75) {};
		\node [style=none] (28) at (3.25, 0) {$\starpart$};
		\node [style=none] (29) at (-1, 0) {$+$};
	\end{pgfonlayer}
	\begin{pgfonlayer}{edgelayer}
		\draw (0.center) to (2.center);
		\draw (1.center) to (3.center);
		\draw (17.center) to (16.center);
		\draw (22.center) to (1.center);
		\draw (22.center) to (23.center);
		\draw (23.center) to (3.center);
		\draw (0.center) to (16.center);
		\draw (2.center) to (17.center);
		\draw (24.center) to (25.center);
		\draw (27.center) to (26.center);
		\draw (27.center) to (24.center);
		\draw (26.center) to (25.center);
	\end{pgfonlayer}
\end{tikzpicture}
        \end{align*}
        To find $\starpart$, we will use a layer-by-layer approach, as in many existing flow-finding algorithms.
        This means we first identify vertices that are maximal in the partial order, and then work `down the order' from there.
        A na\"ive version of the layer-by-layer approach would be as follows.
        \begin{itemize}
         \item Take the system of linear equations $[N_R \mid N_L]$, where $N_R$ are the coefficients and $N_L$ are the attached vectors, i.e.\ the desired values.
         Recall the columns of $N_L$ are labelled by non-output vertices, as are the columns of $\starpart$.

         This means we can alternatively consider $[N_R \mid N_L]$ as a collection of $n-n_O$ independent linear systems, one for each $v\in\comp{O}$, where each system has the same coefficients for the unknowns but generally different constants (and thus different solutions).
         Each of these systems consists of $n-n_O$ equations in $n_O-n_I$ unknowns.
         For most interesting computations, we expect the labelled open graph to be dominated by internal vertices, i.e.\ $V\setminus(I\cup O)$ is larger than $I\cup O$.
         In that case, the linear systems are generally overdetermined.

         Indeed, if all linear systems are solvable, then $N_L + N_R \starpart$ is the all-zero matrix: a trivial DAG.
         Yet, if a flow exists, there must be vertices that are maximal in the associated partial order.
         The columns corresponding to those vertices in the adjacency matrix of the DAG are all-$0$s: thus these vertices must be associated with solvable linear systems.
         \item Perform Gaussian elimination to find the set $L$ of vertices whose associated system of linear equations is solvable. (If none exist, there can be no flow.)
         These are the vertices that are not followed by any not-yet-solved vertices in the partial order represented by the DAG $NC = \rowbasechange{N}{\mathcal{B}} \colbasechange{C}{\mathcal{B}}$.
         \item For each vertex $v\in L$, remove the linear system associated with $v$ from future consideration, i.e.\ ignore the column labelled $v$ in $N_L$ and in $P$.
         Recall that the rows of $N_L$ and $N_R$ are also labelled by non-output vertices.
         Knowing that $v$ is maximal in the partial order among the remaining vertices, we can now remove the row $v$ from consideration in each linear system, eliminating one equation given by $N_R$ coefficients. The removal of the equation means that it is now possible to get a different number than initially required for the equation corresponding to the removed vertex. In other words, it is now possible to add to the partial order a constraint expressing that a vertex solved later precedes some vertex that was solved in earlier steps.

         If there remain systems of linear equations that have not been solved yet, go back to the previous step to look for a new set of vertices whose associated linear systems have become solvable.
        \end{itemize}

        This approach is somewhat similar to the ones in \cite{mhallaFindingOptimalFlows2008a,backensThereBackAgain2021,simmonsRelatingMeasurementPatterns2021} (except that, there, larger matrices must be considered in each step). Since we only have to Gaussian eliminate with respect to the columns of $N_R$ (of which there are $n_O - n_I$), we get an upper bound for the na\"ive approach of $\bigO(n^3 (n_O - n_I))$. However, we can do better by noticing that the matrices constructed for subsequent layers are very similar, and thus we can `reuse' previous Gaussian elimination steps.

        We thus continue with the more involved version of the layer-by-layer algorithm.}

        \item\label{step:LS-and-ILS} \proc{Construct two independent\footnote{By independent, we mean that changes to one do not affect the other.} copies $\LS, \ILS$ of the following matrix.
        \begin{align*}
            \begin{tikzpicture}
	\begin{pgfonlayer}{nodelayer}
		\node [style=none] (17) at (-10.5, 0) {};
		\node [style=none] (18) at (10.5, 0) {};
		\node [style=none] (19) at (0, 4.5) {};
		\node [style=none] (20) at (0, -3.5) {};
		\node [style=none] (21) at (-3.5, 2.5) {};
		\node [style=none] (22) at (-3.5, -2.5) {};
		\node [style=none] (23) at (-5, 0) {$N_R$};
		\node [style=none] (24) at (-6.5, 2.5) {};
		\node [style=none] (25) at (-6.5, -2.5) {};
		\node [style=none] (28) at (-1, 0) {$N_L$};
		\node [style=none] (29) at (1.5, 2.5) {};
		\node [style=none] (30) at (1.5, -2.5) {};
		\node [style=none] (31) at (6.5, 2.5) {};
		\node [style=none] (32) at (6.5, -2.5) {};
		\node [style=none] (33) at (4, 0) {$Id_{\bar{O}}$};
		\node [style=none] (34) at (-6.5, 3) {};
		\node [style=none] (35) at (-3.5, 3) {};
		\node [style=none] (36) at (-5, 3.5) {$n_O-n_I$};
		\node [style=none] (37) at (-3.5, 3) {};
		\node [style=none] (38) at (1.5, 3) {};
		\node [style=none] (39) at (-1, 3.5) {$\bar{O}$};
		\node [style=none] (40) at (1.5, 3) {};
		\node [style=none] (41) at (6.5, 3) {};
		\node [style=none] (42) at (4, 3.5) {$\bar{O}$};
		\node [style=none] (43) at (7, 2.5) {};
		\node [style=none] (44) at (7, -2.5) {};
		\node [style=none] (45) at (8, 0) {$\bar{O}$};
	\end{pgfonlayer}
	\begin{pgfonlayer}{edgelayer}
		\draw (21.center) to (22.center);
		\draw (24.center) to (21.center);
		\draw (24.center) to (25.center);
		\draw (25.center) to (22.center);
		\draw (30.center) to (29.center);
		\draw (22.center) to (30.center);
		\draw (29.center) to (21.center);
		\draw (29.center) to (31.center);
		\draw (31.center) to (32.center);
		\draw (32.center) to (30.center);
		\draw [style=brace edge] (34.center) to (35.center);
		\draw [style=brace edge] (37.center) to (38.center);
		\draw [style=brace edge] (40.center) to (41.center);
		\draw [style=brace edge] (43.center) to (44.center);
	\end{pgfonlayer}
\end{tikzpicture}
        \end{align*}}
        \expl{The first block of the system forms the coefficients of the system of linear equations. The second block consists of the attached vectors, i.e.\ the desired values for which we want to solve the system. Finally, the third block starts as the identity and will be used to keep track of Gaussian elimination steps throughout the process.
        
        $\LS$ stands for `linear system' and $\ILS$ for `initial linear system'.}
        
        \item\label{step:gaussian-eliminate-LS} \proc{Perform Gaussian elimination on $\LS$ (with respect to the first $(n_O - n_I)$ columns).}

        \expl{The idea of the third block, i.e.\ initially the identity matrix, is to know which rows from the initial system $\ILS$ are used in each of the rows of $\LS$, i.e.\ the performed row operations can be `recovered' by reading the third block. The complexity of Gaussian elimination in this step is $\bigO\left(n^2(n_O-n_I)\right)$, as we effectively have to perform Gaussian elimination of an $(n-n_O) \times (n_O-n_I)$ matrix (all rows, but only the first $n_O-n_I$ columns), yet the cost of each row operation is $(n_O-n_I)+(n-n_O)+(n-n_O) = 2n-n_O-n_I \in \bigO(n)$, as this is the actual size of each row.}

        \item\label{step:initialise-L-P} \proc{Define the (initially empty) set of already solved vertices $\solved = \emptyset$ and initialize the matrix of the solutions $\starpart$ as an (initially empty) $(n_O-n_I) \times (n - n_O)$ matrix with columns corresponding to non-outputs $\comp{O}$.}

        \item\label{step:while-loop} \proc{Repeat the following loop of substeps until $\solved = \comp{O}$, i.e.\ all vertices are solved.
        We will refer to this loop as the `\texttt{while} loop'.}

        \expl{Each run of the \texttt{while} loop corresponds to solving a single layer of vertices. The loop invariant is that the first $(n_O-n_I)$ columns of $\LS$ are in row echelon form (but not necessarily row-reduced echelon form).
        The loop invariant holds at the beginning by Step~\ref{step:gaussian-eliminate-LS}.}
        
        \begin{enumerate}
            \item\label{step:find-L} \proc{Identify a set $\tosolve$ of vertices $v \in \comp{O} \setminus \solved$ whose associated system of linear equations can be solved at this step.
            \begin{itemize}
                \item If the set $\tosolve$ turns out to be empty, return that there is no flow.
            \end{itemize}}
            \expl{This can be done in $\bigO(n^2)$, as $\LS$ is in row echelon form:\begin{itemize}
                \item The system of linear equations associated with a vertex $v$ can be solved if it has not been solved yet and if its column in the second block (which determines the constants in each equation) has only zeros where it intersects rows for which all the coefficients in the first block are $0$s:
                \begin{align*}
                    \begin{tikzpicture}
	\begin{pgfonlayer}{nodelayer}
		\node [style=none] (0) at (-10, 0) {};
		\node [style=none] (1) at (10, 0) {};
		\node [style=none] (2) at (0, 4.5) {};
		\node [style=none] (3) at (0, -3.5) {};
		\node [style=none] (4) at (-3.5, 2.5) {};
		\node [style=none] (5) at (-3.5, -2.5) {};
		\node [style=none] (6) at (-6.5, 2.5) {};
		\node [style=none] (7) at (-6.5, -2.5) {};
		\node [style=none] (8) at (1.5, 2.5) {};
		\node [style=none] (9) at (1.5, -2.5) {};
		\node [style=none] (10) at (6.5, 2.5) {};
		\node [style=none] (11) at (6.5, -2.5) {};
		\node [style=BlackTEXT] (12) at (4, 0) {third\\ block};
		\node [style=none] (14) at (-6.5, -0.75) {};
		\node [style=none] (15) at (-3.5, -0.75) {};
		\node [style=none] (16) at (1.5, -0.75) {};
		\node [style=none] (17) at (-0.5, 3) {$v$};
		\node [style=BlackTEXT] (18) at (-5, 3.5) {Row echelon};
		\node [style=none] (19) at (-6.5, 3) {};
		\node [style=none] (20) at (-3.5, 3) {};
		\node [style=none] (21) at (-6, 2) {$1$};
		\node [style=none] (22) at (-5.5, 2) {$*$};
		\node [style=none] (23) at (-5, 2) {$*$};
		\node [style=none] (24) at (-4.5, 2) {$*$};
		\node [style=none] (25) at (-4, 2) {$*$};
		\node [style=none] (26) at (-6, 1.25) {$0$};
		\node [style=none] (27) at (-5.5, 1.25) {$1$};
		\node [style=none] (28) at (-5, 1.25) {$*$};
		\node [style=none] (29) at (-4.5, 1.25) {$*$};
		\node [style=none] (30) at (-4, 1.25) {$*$};
		\node [style=none] (31) at (-6, 0.5) {$0$};
		\node [style=none] (32) at (-5.5, 0.5) {$0$};
		\node [style=none] (33) at (-5, 0.5) {$0$};
		\node [style=none] (34) at (-4.5, 0.5) {$1$};
		\node [style=none] (35) at (-4, 0.5) {$*$};
		\node [style=none] (36) at (-6, -0.25) {$0$};
		\node [style=none] (37) at (-5.5, -0.25) {$0$};
		\node [style=none] (38) at (-5, -0.25) {$0$};
		\node [style=none] (39) at (-4.5, -0.25) {$0$};
		\node [style=none] (40) at (-4, -0.25) {$1$};
		\node [style=none] (41) at (-5, -1.5) {$\scalebox{3}{\raisebox{-1.75ex}{0}}$};
		\node [style=none] (42) at (-0.5, 2) {$*$};
		\node [style=none] (43) at (-0.5, 1.25) {$*$};
		\node [style=none] (44) at (-0.5, 0.5) {$*$};
		\node [style=none] (45) at (-0.5, -0.25) {$*$};
		\node [style=none] (46) at (-0.5, -1) {$\raisebox{-2ex}{0}$};
		\node [style=none] (47) at (-0.5, -1.5) {$\raisebox{-2.25ex}{0}$};
		\node [style=none] (48) at (-0.5, -2) {$\raisebox{-2.5ex}{0}$};
	\end{pgfonlayer}
	\begin{pgfonlayer}{edgelayer}
		\draw (6.center) to (4.center);
		\draw (6.center) to (7.center);
		\draw (7.center) to (5.center);
		\draw (8.center) to (4.center);
		\draw (8.center) to (10.center);
		\draw (10.center) to (11.center);
		\draw (11.center) to (9.center);
		\draw (14.center) to (15.center);
		\draw (4.center) to (15.center);
		\draw (8.center) to (16.center);
		\draw [style=empty] (5.center)
			 to (15.center)
			 to (16.center)
			 to (9.center)
			 to cycle;
		\draw (15.center) to (16.center);
		\draw (16.center) to (9.center);
		\draw (9.center) to (5.center);
		\draw (5.center) to (15.center);
		\draw [style=brace edge] (19.center) to (20.center);
	\end{pgfonlayer}
\end{tikzpicture}
                \end{align*}
                Here, $*$ denotes matrix entries of arbitrary value.

                \item This means, the submatrix marked in grey must be checked for all-$0$ columns. Checking this submatrix takes $\bigO(n^2)$ as its size is bounded above by $n \times n$.
            \end{itemize}
            
            If there remain vertices under consideration but we cannot solve any of the associated systems of linear equations, then Pauli flow cannot exist.}

            \item\label{step:solve-L} \proc{Solve the linear systems associated with all vertices $v \in \tosolve$.}

            \expl{Since the matrix is in row echelon form, finding a solution for the linear system associated with $v$ is fast and is done by backtracking over the coefficient matrix (i.e.\ the first block) and the column of the solved vertex.}
            
            \item\label{step:record-solution-in-P} \proc{For each $v\in L$, place the solution to the system of linear equations associated with $v$ in the $v$ column of $\starpart$.}

            \expl{Solving the system of linear equations associated with one vertex takes at most $\bigO\left(n(n_O-n_I)\right)$ steps due to row echelon form, as there are only $n_O-n_I$ columns of the coefficient block. Hence the cost per layer (i.e.\ per round of the loop) of this step is $\bigO\left(ns(n_O-n_I)\right)$, where $s$ is the number of vertices in the layer, i.e.\ $s = |\tosolve|$.
            }

            \item\label{step:for-loop} \proc{Bring the linear system $\LS$ to the row echelon form that would be achieved by Gaussian elimination if the row and column vectors corresponding to vertices in $\tosolve$ where not included in the starting matrix.}

            \expl{This is the step where we remove certain rows and columns from future consideration.
            In the case of columns, these are simply columns of constants: we ignore those systems of linear equations that have already been solved.
            In the case of rows, this involves setting an entire row of coefficients to 0 in a specific way.
            This step may therefore break row echelon form.

            To remove the rows and columns, we iterate over all vertices in $\tosolve$; we will refer to this inner loop as the `\texttt{for} loop'. Each vertex appears in $L$ during at most one round of the outer \texttt{while} loop. This means that we can spend $\bigO(n^2)$ operations per vertex on modifying the set of systems of linear equations, while staying within the overall limit of $\bigO(n^3)$.

            Note that if one layer (i.e.\ one round of the \texttt{while} loop) contains a substantial part of $\comp{O}$, then working out this layer in the \texttt{while} loop may take a number of operations which is not in $\bigO(n^2)$. However, this does not increase the overall complexity to $\bigO(n^4)$, as the cost for each vertex is still $\bigO(n^2)$ and each vertex appears only once.

            The process of modifying the systems of linear equations is performed as follows.}\begin{enumerate}
            
                \item\label{step:update-S} \proc{Let $v \in \tosolve$ be the vertex currently considered in the \texttt{for} loop. Add $v$ to the set $\solved$.}

                \expl{The addition to set $\solved$ ensures that a vertex will not be solved for a second time. Next, we must remove the $v$ row from all the systems of linear equations.}
                
                \item\label{step:find-v-dependent-rows} \proc{Find all the rows of the linear system which contain a $1$ in the $v$ column of the third block, denote them by $r_1, r_2, \dots, r_k$.}

                \expl{These are the rows that depend on the original $v$-labelled row.
                They can be identified by iterating over the $v$ column in the third block, which gives a complexity of $\bigO(n)$.}
                
                \item\label{step:row-additions-r1-rk-1} \proc{Add row $r_k$ to rows $r_1, r_2, \dots, r_{k-1}$.}

                \expl{This removes the dependence on the original $v$-labelled row from $r_1, r_2, \dots, r_{k-1}$, meaning only $r_k$ now depends on the original $v$ row.

                By the loop invariant, the system is in row echelon form initially. Adding a later row to earlier rows preserves row echelon form.

                The additions take $\bigO(nk)$ steps, as there are at most $\bigO(n)$ columns in the system, and $k-1 \in \bigO(k)$ row operations are performed. Since $k$ is bounded by the number of rows, we get $\bigO(nk) \in \bigO(n^2)$. At the end, the matrix is in row echelon form, and only row $r_k$ uses the original row of vertex $v$ anywhere.}

                \item\label{step:row-addition-rk} \proc{Take the $v$-labelled row from $\ILS$, the initial linear system, and add this row to $r_k$.}

                \expl{In this way, we remove the dependency of $r_k$ on the initial row of $v$.

                This operation can break row echelon form. Crucially, however, only $r_k$ can break the row echelon form. We can correct this one row in the following step.}

                \item\label{step:simplify-rk} \proc{Add other rows of the current linear system $\LS$ to row~$r_k$ to simplify the latter as much as possible.}

                \expl{In particular, the row $r$ is added to $r_k$ when $r_k$ has a $1$ in the column corresponding to the leading $1$ of row\footnote{It suffices to only reduce $r_k$ until a $1$ that cannot be eliminated appears in $r_k$, or until $r_k$ becomes identically $0$, whichever is first.} $r$. These are the same operations that one would perform in Gaussian elimination, except that we do not use $r_k$ itself to cancel 1s in other rows, we only use other rows to cancel 1s in $r_k$. This is because we only need the coefficient block to be in row echelon form (not row-reduced echelon form).

                Hence, we need to perform at most $n_O-n_I$ row operations, as this is the maximum possible number of rows that would need to be added to row $r_k$. Hence, this step is bounded by $\bigO\left(n(n_O-n_I)\right)$.}
                
                \item\label{step:swap-rk} \proc{Swap the rows as necessary to place what used to be the row $r_k$ in the correct spot to get a row echelon form for the coefficient block.}

                \expl{There are at most $n_O-n_I$ non-zero rows in the coefficient block of the matrix. This is due to row-echelon form and the rank of the coefficient submatrix being bounded by its number of columns.
                Moreover, at most one swap will involve a row whose coefficient part is identically $0$: if $r_k$ is not identially $0$ but it is initially in the block of all-$0$ rows, then a single swap suffices to place it adjacent to the block of non-trivial rows.
                Whether $r_k$ is identically $0$ or not, all subsequent swaps are done only between non-trivial rows and row $r_k$. Hence, at most $\bigO\left(n(n_O-n_I)\right)$ operations are required. After that, the $\LS$ is in the row echelon form that would be achieved if there was no $v$ row in the initial system (or, more precisely, if the $v$ row of the initial system was identically $0$).}
            \end{enumerate}

            Combining the steps above and assuming a layer of size $s$, the total run time of the process for modifying the system of linear equations is $\bigO(n^2s)$, as desired.
        \end{enumerate}
        
        Since each vertex appears in at most one layer (i.e.\ one iteration of the \texttt{while} loop), and each substep has a complexity at most $\bigO(n^2 s)$, the entire \texttt{while} loop runs in total time $\bigO(n^3)$.

        \item\label{step:build-c-B} \proc{Construct $\colbasechange{C}{\mathcal{B}}$ by stacking $Id_{\comp{O}}$ over $\starpart$.}

        \expl{Once the \texttt{while} loop has been completed, we know that all vertices have been solved. Thus, a matrix $\starpart$ such that $N_L + N_R \starpart$ is a DAG has been found, and we get $\colbasechange{C}{\mathcal{B}} = \left[ \frac{Id_{\comp{O}}}{\starpart} \right]$, cf.\ the illustration in \eqref{eq:image-c-b}.}

        \item\label{step:output} \proc{Return the tuple $C'\colbasechange{C}{\mathcal{B}}, NC'\colbasechange{C}{\mathcal{B}}$ as the Pauli flow.}

        \expl{We return the correction function $c$ in matrix form and the relation $\trl_c$, also in matrix form, cf.\ Step~\ref{step:basis-change} for why the matrices above are the correct ones. Computing and outputting the matrix products takes at most $\bigO(n^3)$ operations, assuming the standard matrix multiplication algorithm is used.}
    \end{enumerate}

    As previously, if one requires the partial order $\prec_c$, then the transitive closure of $\trl_c$ must be computed, which can be done in $\bigO(n^2)$ steps. All steps together give the desired bound of $\bigO(n^3)$.
\end{proof}

We provide pseudocode for this algorithm in Appendix~\ref{pseudocodes hard}.
A worked example of running the algorithm on the labelled open graph from Figure~\ref{fig:loG} is given in Appendix~\ref{run example}.

With the above procedure, we proved that Pauli flow on a labelled open graph can be found in $\bigO(n^3)$ time, where $n$ is the number of vertices in the underlying graph. However, extra care in bounding the complexity of individual steps could result in a slightly better overall bound: The only situation in which a round of the inner \texttt{for} loop (cf.\ Step~\ref{step:for-loop}) may actually require $\bigO(n^2)$ operations, rather than $\bigO(n(n_O-n_I))$ operations, occurs when there are many rows which depend on a vertex that has just been solved, i.e.\ when the value $k$ in Step~\ref{step:find-v-dependent-rows} is large. Therefore, it is possible that with extra care one could show only $\bigO\left( n^2 (n_O-n_I) \right)$ operations are needed for the entire \texttt{while} loop of Step~\ref{step:while-loop}.
Then, if there are methods faster than $\bigO(n^3)$ for finding the right inverse, kernel, and product of $\bigO(n) \times \bigO(n)$ matrices, a better bound than $\bigO(n^3)$ could be obtained for the entire algorithm even in the case of $n_I < n_O$.

Another possible optimisation is to notice that the order-demand matrix $N$ has many rows which are identically $0$. After the basis change to $\rowbasechange{N}{\mathcal{B}}$, these rows are still identically $0$. Therefore, removing them from the matrix at the start is beneficial: then, the algorithm has to consider a matrix with fewer rows. The rows that are always identically $0$ in the order-demand matrix are those of planar-measured vertices and $XY$-measured inputs. Hence, the number of rows of $N$ can be reduced from $n-n_O$ to $\left| \measplanar \setminus I \right|$.

Finally, we point out that finding the solutions for the individual systems of linear equations associated with the vertices in $\tosolve$ and the removal of the rows corresponding to these vertices does not have to be split into two separate loops. However, without breaking these steps into two loops, an insufficiently careful method of finding the solution could impose an order relation $v \prec u$ on two vertices from the same layer\footnote{In other words: while the flow returned by our algorithm is guaranteed to have the useful property of being `maximally delayed' \cite{mhallaFindingOptimalFlows2008a,backensThereBackAgain2021,simmonsRelatingMeasurementPatterns2021}, such a change could break that guarantee.}.

\section{Summary and future work}\label{sec:summary and future work}

We have provided a simplified interpretation of Pauli flow and improved flow-finding algorithms. Given a labelled open graph $\Gamma$, we introduced the notions of a `flow-demand matrix' $M$ and an `order-demand matrix' $N$. We proved that the existence of a Pauli flow on $\Gamma$ corresponds to the existence of right inverse $C$ of $M$, such that $NC$ forms the adjacency matrix of a DAG. This provides a new algebraic interpretation of Pauli flow, thus extending previously known algebraic interpretations for $XY$ only measurements \cite{mhallaWhichGraphStates2014a}, and $X$ and $Z$ only measurements \cite{mitosekPauliFlowOpen2024}.

Using our new algebraic interpretation, we developed two $\bigO(n^3)$ algorithms for finding flow. The algorithm for the general case works for any number of inputs and outputs and is an improvement of the layer-by-layer approach used in \cite{mhallaFindingOptimalFlows2008a,backensThereBackAgain2021,simmonsRelatingMeasurementPatterns2021}.
In the special case of equal numbers of inputs and outputs, the algebraic interpretation gives a new and simpler proof that a focused Pauli flow (if it exists) must be unique.
In this case, there is a simpler flow-finding algorithm, which differs from other existing flow-finding algorithms\footnote{A version for $XY$-only measurements could be derived from \cite{mhallaWhichGraphStates2014a} but is not suggested there -- possibly because it gives no improvement on the complexity of the previously-known algorithm for finding $XY$-only gflow \cite{mhallaFindingOptimalFlows2008a}.}. Further, we reduced the problem of finding flow on labelled open graphs with equal numbers of inputs and outputs to, and from, the problem of finding the inverse of a matrix and performing matrix multiplication. We thus argue that further improvements to flow-finding algorithms must necessarily lead to or come from new methods for these standard linear algebra problems.

Furthermore, in the case of an equal number of inputs and outputs, we showed that the Pauli flow can always be reversed, extending the proof from \cite{mhallaWhichGraphStates2014a}.
When the number of inputs and outputs does not agree, we showed that the space of focused sets of a Pauli flow -- i.e.\ the sets of vertices that have a trivial net correction effect --  is isomorphic to the kernel of the flow-demand matrix. Hence, we obtained an interpretation of the space of focused sets different to the one from \cite{simmonsRelatingMeasurementPatterns2021}.

The algebraic interpretation and our new algorithms both contribute to further theoretical understanding of Pauli flow and offer many applications.

\subsection{Future work}

In the remaining part, we discuss ideas for future work.

\paragraph{Other notions of MBQC} In the introduction, we outlined the correspondence between Pauli flow and \textit{strong uniform stepwise deterministic} (also called robustly deterministic) quantum computation, the standard notion of determinism for the one-way model. However, other relevant notions of determinism exist, for instance, see \cite{mhallaWhichGraphStates2014a}. It may be interesting to check if algebraic interpretations can be established for those notions of determinism as well.

Similarly, there are other types of flow defined for different fragments of quantum computation. For instance in the case of MBQC schemes over qudits of odd prime dimension (instead of qubits), a `$Zd$ flow' was introduced as a necessary and sufficient condition for strong determinism \cite{boothOutcomeDeterminismMeasurementbased2023}. An $\bigO(n^4)$ algorithm for finding such flow was given in the same paper. We speculate that an algebraic interpretation could exist for $Zd$ flow and may lead to an $\bigO(n^3)$ algorithm.

\paragraph{Pauli flow on open graphs with unknown measurement labels} In \cite{mitosekPauliFlowOpen2024}, one of the authors of this paper showed that given an (unlabelled) open graph one could decide the existence of a measurement labelling resulting in the labelled open graph having Pauli flow. This was done using an early (and unpublished) version of our algebraic interpretation for the special case of $X$ and $XY$ measurements. Additionally, a different version of the algebraic interpretation, applying to $X$ and $Z$ measurements only, was used to show that the problem of finding measurement labellings compatible with Pauli flow is in the complexity class $\RM{RP}$ (random polynomial time).

Here, we improved the algebraic interpretation by providing a different way of dealing with $Z$-measured vertices. We hope this new version of algebraic interpretation could lead to a complexity class $\RM{P}$ for deciding the existence of a measurement labelling resulting in the labelled open graph having Pauli flow.

\paragraph{Circuit extraction} It is straightforward to translate a quantum circuit into a robustly deterministic one-way model computation (for example, see \cite{broadbentParallelizingQuantumCircuits2009}).
For the reverse process, it is straightforward to translate a one-way computation with causal flow into a circuit, or to translate a one-way computation with some more general flow property into a quantum circuit with $\bigO(n)$ ancillas, where $n$ is the number of qubits in the original computation.
The problem of translating a unitary one-way computation with extended gflow or Pauli flow into an ancilla-free quantum circuit (i.e.\ a circuit with the minimal number $n_O-n_I$ of qubit initialisations) has proved more complicated.

It was finally resolved over recent years by a series of papers using the ZX calculus, a graphical language for quantum computation \cite{coeckeInteractingQuantumObservables2011, coeckePicturingQuantumProcesses2017, vandeweteringZXcalculusWorkingQuantum2020}. One of the essential problems in ZX calculus is \emph{circuit extraction}, i.e.\ the problem of translation from a diagram to an equivalent quantum circuit. In general, this problem is $\counting[P]$-hard \cite{debeaudrapCircuitExtractionZXDiagrams2022c}. However, if the ZX diagram can be interpreted as a labelled open graph and exhibits a flow property such as Pauli flow, then polynomial time algorithms for circuit extraction were found \cite{backensThereBackAgain2021, simmonsRelatingMeasurementPatterns2021, staudacherMulticontrolledPhaseGate2024}.

The method from \cite{backensThereBackAgain2021} for gflow only requires that the flow exists, but it does not need to find it. However, the Pauli-flow circuit extraction \cite{simmonsRelatingMeasurementPatterns2021} algorithm first needs to find a Pauli flow and then performs operations to transform a labelled open graph to a quantum circuit. Our algorithms improve the first part of this extraction procedure. Therefore, it would be interesting to check whether the new flow-finding algorithm leads to new state-of-the-art methods for circuit extraction, both from MBQC schemes and ZX diagrams.

\paragraph{Circuit optimization} One of the most important uses of ZX calculus is circuit optimization (for example, see\cite{duncanGraphtheoreticSimplificationQuantum2020, staudacherReducing2QuBitGate2023}). An unoptimized diagram can be transformed into ZX diagram, then simplified with ZX rewrites, and finally, a new optimized circuit can be extracted. As outlined above, this method requires the final diagram to exhibit Pauli flow. To ensure this, current methods start with a diagram exhibiting Pauli flow and maintain Pauli flow in each rewrite step. The tricky part is which rewrite rules to apply: They must both preserve the existence of flow, and lead to a \textit{better} circuit in the end -- here, \textit{better} can mean different things, such as better $T$-count, a smaller number of $2$-qubit gates, or another metric. Rewrites rules preserving flow have been studied in, for example, \cite{mcelvanneyCompleteFlowpreservingRewrite2022a, mcelvanneyFlowpreservingZXcalculusRewrite2023a}. The current process of finding flow-preserving rewrite rules is tedious, as it requires verification of all Pauli flow conditions by hand. Our new algebraic interpretation may simplify this process and hence lead to the discovery of new flow-preserving rewrite rules.

\section*{Acknowledgements}
The authors thank Korbinian Staudacher for the helpful discussions about Pauli flow. PM also thanks Michał Łupiński for explaining relevant technical aspects of linear algebra.

\phantomsection

\addcontentsline{toc}{chapter}{Bibliography}

\begin{sloppypar}
\bibliographystyle{eptcs}
\bibliography{librefwin}
\end{sloppypar}

\newpage
\appendix

\section{Pseudocode for matrix constructions and flow-finding in case of equal numbers of inputs and outputs}\label{pseudocodes easy}

The first algorithm relates to Definitions~\ref{flow-demand matrix definition} and~\ref{order-demand matrix definition}.
The second algorithm relates to Theorem~\ref{algo easy}.

\algrenewcommand\algorithmicdo{}
\algrenewcommand\algorithmicthen{}
\begin{algorithm}[H]
\caption{Construction of flow and order-demand matrices}\label{M and N pseudo}
\begin{algorithmic}[1]
\Statex Returns the flow-demand matrix of a given labelled open graph
\Procedure{Flow-DemandMatrix}{$G,I,O,\lambda$}
    \State $M = {Adj}_G \mid_{\comp{I}}^{\comp{O}}$ \Comment{Construction of reduced adjacency matrix}
    \For {$v \in \comp{O}$} \Comment{For all measured vertices}
        \If {$\lambda(v) \in \{Z, YZ, XZ \}$}
            \State $M_{v,*} \mathrel{*}= 0$ \Comment{Multiply rows of $Z, YZ, XZ$ measurements by $0$}
        \EndIf
        \If {$\lambda(v) \in \{ Y, Z, YZ, XZ \}  \wedge v \notin I$}
            \State $M_{v,v} = 1$ \Comment{Set intersections of rows and columns for $Y, Z, YZ, XZ$ measurements to $1$}
        \EndIf
    \EndFor
    \State \Return $M$
\EndProcedure
\[\]
\Statex Returns the order-demand matrix of a given labelled open graph
\Procedure{Order-DemandMatrix}{$G,I,O,\lambda$}
    \State $N = {Adj}_G \mid_{\comp{I}}^{\comp{O}}$ \Comment{Construction of reduced adjacency matrix}
    \For {$v \in \comp{O}$}
        \If {$\lambda(v) \in \{ X, Y, Z, XY \}$}
            \State $N_{v,*} \mathrel{*}= 0$ \Comment{Multiply rows of $XY$ and Pauli measurements by $0$}
        \EndIf
        \If $\lambda(v) \in \{ XY, XZ \} \wedge v \notin I$
            \State $N_{v,v} = 1$ \Comment{Set intersections of rows and columns for $XY, XZ$ measurements to $1$}
        \EndIf
    \EndFor
    \State \Return $N$
\EndProcedure
\end{algorithmic}
\end{algorithm}

\begin{algorithm}[H]
\caption{Finding flow when $n_I = n_O$}\label{nI equal nO algo}
\begin{algorithmic}[1]
\Statex Checks if a labelled open graph $\Gamma = (G,I,O,\lambda)$ with $n_I = n_O$ has Pauli flow and, if it exists, returns the focused extensive correction function $c$ (in the form of a matrix) and the corresponding $\trl_c$ (also in the form of a matrix)
\Procedure{FindFlowSimple}{$G,I,O,\lambda$}
    \State $M = \Call{Flow-DemandMatrix}{G,I,O,\lambda}$
    \State $N = \Call{Order-DemandMatrix}{G,I,O,\lambda}$
    \State $C = \Call{Inverse}{M}$
    \If {$C$ is $\mathbf{None}$}
        \State \Return `NO FLOW EXISTS'
    \EndIf
    \State $R = NC$
    \If {$R$ is not a DAG}
        \State \Return `NO FLOW EXISTS'
    \EndIf
    \State \Return $(C, R)$
\EndProcedure
\end{algorithmic}
\end{algorithm}

\section{Pseudocode for case of different numbers of inputs and outputs}\label{pseudocodes hard}

This algorithm relates to Theorem~\ref{algo hard}. We list the corresponding steps from Theorem~\ref{algo hard} in the comments. \textit{GE} indicates that the step can be performed using Gaussian elimination.

\begin{algorithm}
\caption{Finding flow in the general case}\label{general algo}
\begin{algorithmic}[1]
\Statex Checks if a labelled open graph $\Gamma = (G,I,O,\lambda)$ has Pauli flow and, if it exists, returns the focused extensive correction function $c$ (in the form of a matrix) and the corresponding $\trl_c$ (also in the form of a matrix)
\Procedure{FindFlowGeneral}{$G,I,O,\lambda$}
    \State $M = \Call{Flow-DemandMatrix}{G,I,O,\lambda}$ \Comment{Step~\ref{step:M}}
    \State $N = \Call{Order-DemandMatrix}{G,I,O,\lambda}$ \Comment{Step~\ref{step:N}}
    \If {$\rank M \ne n-n_O$} \Comment{GE; Step~\ref{step:M right-invertibility}}
        \State \Return `NO FLOW EXISTS'
    \EndIf
    \State $C_0 = \text{any right inverse of $M$}$ \Comment{GE and backtracking; Step~\ref{step:find-inverse}}
    \State $F = \text{matrix with columns forming basis of $\ker M$}$ \Comment{GE; Step~\ref{step:find-kernel}}
    \State $C' = \left[ C_0 \mid F \right]$ \Comment{Step~\ref{step:c'}}
    \State $\rowbasechange{N}{\mathcal{B}} = NC'$ \Comment{Step~\ref{step:basis-change}}
    \State $N_L = \text{first $n-n_O$ columns of $\rowbasechange{N}{\mathcal{B}}$}$ \Comment{Step~\ref{step:NL-and-NR}}
    \State $N_R = \text{last $n_O-n_I$ columns of $\rowbasechange{N}{\mathcal{B}}$}$
    \State $\ILS, \LS = \left[ N_R \mid N_L \mid Id_{\comp{O}} \right]$ \Comment{Two independent copies of the same linear system; Step~\ref{step:LS-and-ILS}}
    \Statex \Comment{From now, we refer to the three parts of $\ILS$ and $\LS$ as the first block, the second block, and the third block}
    \State $\text{bring first block of $\LS$ into row echelon form}$ \Comment{GE; Step~\ref{step:gaussian-eliminate-LS}}
    \State $\solved = \emptyset$ \Comment{Step~\ref{step:initialise-L-P}}
    \State $\text{initialize $(n_O-n_I) \times (n-n_O)$ matrix $\starpart$}$ \Comment{Columns of $P$ correspond to non-outputs $\comp{O}$}
    \While {$\solved \ne \comp{O}$}: \Comment{Step~\ref{step:while-loop}}
        \State $\text{find the first row $r_z$ in $\LS$ whose first $n_O-n_I$ entries equal $0$}$ \Comment{Step~\ref{step:find-L}}
        \State $\tosolve = \{ v \in \comp{O} \setminus \solved \mid \text{column $v$ in the first $\LS$ block has all entries from row $r_z$ on equal to $0$} \}$
        \If {$\tosolve = \emptyset$}
            \State \Return `NO FLOW EXISTS'
        \EndIf
        \For {$v \in \tosolve$}: \Comment{Steps~\ref{step:solve-L} and \ref{step:record-solution-in-P}}
            \State $\starpart_{*,v} = \Call{SolveLinearSystem}{\text{first block of $\LS$, $v$ column in second block of $\LS$}}$
        \EndFor
        \For {$v \in \tosolve$}: \Comment{Check text after Theorem~\ref{algo hard} for why we need two loops; Step~\ref{step:for-loop}}
            \State $\solved = \solved \cup \left\{ v \right\}$ \Comment{Step~\ref{step:update-S}}
            \State $R = \left[ r \mid \text{$r$ is a row whose intersection with the $v$ column in the third block of $\LS$ is $1$} \right]$
            \Statex \Comment{$R$ is an ordered list that corresponds to $r_1, r_2, \dots, r_k$ in Step~\ref{step:find-v-dependent-rows}}
            \State $r_{last} = \text{last element of $R$}$ \Comment{corresponds to $r_k$}
            \For {$r \in R[:-1]$} \Comment{Iterate over all but last element; Step~\ref{step:row-additions-r1-rk-1}}
                \State $\text{add $r_{last}$ of $\LS$ to row $r$ in $\LS$}$
            \EndFor
            \State $\text{add row $v$ of $\ILS$ to row $r_{last}$ in $\LS$}$ \Comment{Step~\ref{step:row-addition-rk}}
            \For $r \in \text{rows of $\LS$ except $r_{last}$}$: \Comment{Iterate from top; the loop covers Step~\ref{step:simplify-rk}}
                \If {$\text{row $r$ of $\LS$ has first $n_O-n_I$ entries $0$}$}
                    \State Break
                \EndIf
                \State $y = \text{column of the leading $1$ in row $r$ of $\LS$}$
                \If {$\text{intersection of $r_{last}$ and column $y$ in $\LS$ contains $1$}$}
                    \State $\text{add row $r$ to $r_{last}$ in $\LS$}$
                \EndIf
            \EndFor
            \State swap $r_{last}$ with other rows to bring first block of $\LS$ back into row echelon form
            \Statex \Comment{Step~\ref{step:swap-rk}}
        \EndFor
    \EndWhile
    \State $\colbasechange{C}{\mathcal{B}} = \begin{bmatrix}
        Id_{\comp{O}}\\
        \starpart
    \end{bmatrix}$ \Comment{Step~\ref{step:build-c-B}}
    \State \Return $\left(C' \colbasechange{C}{\mathcal{B}}, N C'\colbasechange{C}{\mathcal{B}}\right)$ \Comment{Step~\ref{step:output}}
\EndProcedure
\end{algorithmic}
\end{algorithm}

\FloatBarrier
\section{Worked example of the Pauli-flow finding algorithm for more outputs than inputs}\label{run example}

This section contains a worked example of the complicated parts of the algorithm of Theorem~\ref{algo hard}, using the running example from the main body of the paper.

Consider the labelled open graph $\Gamma$ from Figure~\ref{fig:loG}. The flow-demand matrix $M$ and order-demand matrix $N$ are presented in Figure~\ref{fig:mats}.
Suppose we do not yet know the correction function of Figure~\ref{fig:pf} or the matrix $C$ from Figure~\ref{fig:C}.
Instead, suppose we have found some right inverse $C_0$ of $M$ as well as a matrix $F$ whose columns form a basis of $\ker M$.
At Step~\ref{step:c'} of the algorithm, these are composed into the matrix $C' = \left[ C_O \mid F \right]$ shown below:\begin{gather*}
    C' = [C_0 \mid F] = \begin{array}{|c|ccccc||c|}
        \toprule
         &   i& a& b& e& d& F_1\\
        \midrule
        a&   0& 1& 0& 0& 0& 0\\
        b&   1& 0& 0& 0& 0& 0\\
        e&   1& 0& 1& 0& 1& 1\\
        d&   0& 0& 0& 0& 1& 0\\
        o_1& 1& 0& 0& 1& 1& 1\\
        o_2& 0& 0& 0& 0& 0& 1\\
        \bottomrule
    \end{array}
\end{gather*}

In general, $F$ will have $n_O - n_I$ columns. Here, this number is one. In the following, we use $F_1$ as the label for the single column of the matrix $F$.

Then, for Step~\ref{step:basis-change}, the matrices $\rowbasechange{M}{\mathcal{B}} = MC'$ and $\rowbasechange{N}{\mathcal{B}} = NC'$ are as follows\footnote{The matrix $\rowbasechange{M}{\mathcal{B}}$ is not required for the algorithm, we provide it here for completeness.}:\begin{gather*}
    \begin{array}{|c|cccccc|}
        \toprule
        \rowbasechange{M}{\mathcal{B}}& i& a& b& e& d& F_1\\
        \midrule
        i& 1& 0& 0& 0& 0  & 0\\
        a& 0& 1& 0& 0& 0  & 0\\
        b& 0& 0& 1& 0& 0  & 0\\
        e& 0& 0& 0& 1& 0  & 0\\
        d& 0& 0& 0& 0& 1  & 0\\
        \bottomrule
    \end{array}
    \qquad
    \begin{array}{|c|cccccc|}
        \toprule
        \rowbasechange{N}{\mathcal{B}}& i& a& b& e& d& F_1\\
        \midrule
        i& 0& 0& 0& 0& 0& 0\\
        a& 0& 1& 0& 0& 1& 1\\
        b& 0& 0& 0& 0& 0& 0\\
        e& 1& 0& 1& 0& 1& 1\\
        d& 0& 0& 0& 0& 0& 0\\
        \bottomrule
    \end{array}
\end{gather*}

In Steps~\ref{step:NL-and-NR} and~\ref{step:LS-and-ILS}, we break $\rowbasechange{N}{\mathcal{B}}$ into $N_L$ and $N_R$, and obtain the following linear systems $\LS, \ILS$. We give columns from the third block primed labels, to distinguish them from the columns in the second block:\begin{gather*}
    \ILS = \LS = \begin{array}{|c|c|ccccc|ccccc|}
    \toprule
     & F_1 & i & a & b & e & d & i' & a' & b' & e' & d' \\
    \midrule
    i & 0 & 0 & 0 & 0 & 0 & 0 & 1 & 0 & 0 & 0 & 0 \\
    a & 1 & 0 & 1 & 0 & 0 & 1 & 0 & 1 & 0 & 0 & 0 \\
    b & 0 & 0 & 0 & 0 & 0 & 0 & 0 & 0 & 1 & 0 & 0 \\
    e & 1 & 1 & 0 & 1 & 0 & 1 & 0 & 0 & 0 & 1 & 0 \\
    d & 0 & 0 & 0 & 0 & 0 & 0 & 0 & 0 & 0 & 0 & 1 \\
    \bottomrule
    \end{array}
\end{gather*}

In Step~\ref{step:gaussian-eliminate-LS}, we apply Gaussian elimination to the first $n_O - n_I$ columns, i.e.\ in this case to the first column only. This means adding the second row to the fourth and then swapping the first two rows. We drop previous row labels, as they no longer make sense after performing row operations. (Nevertheless, the relationship to the original vertex-labelled rows continues to be encoded in the third block.) We obtain the following updated $\LS$ linear system:\begin{equation}\label{eq:LS-first-row-echelon}
    \begin{array}{|c|c|ccccc|ccccc|}
    \toprule
     & F_1 & i & a & b & e & d & i' & a' & b' & e' & d' \\
    \midrule
    1 & 1 & 0 & 1 & 0 & 0 & 1 & 0 & 1 & 0 & 0 & 0 \\
    2 & 0 & 0 & 0 & 0 & 0 & 0 & 1 & 0 & 0 & 0 & 0 \\
    3 & 0 & 0 & 0 & 0 & 0 & 0 & 0 & 0 & 1 & 0 & 0 \\
    4 & 0 & 1 & 1 & 1 & 0 & 0 & 0 & 1 & 0 & 1 & 0 \\
    5 & 0 & 0 & 0 & 0 & 0 & 0 & 0 & 0 & 0 & 0 & 1 \\
    \bottomrule
    \end{array}
\end{equation}

Importantly, $\ILS$ remains unchanged. Next, in Step~\ref{step:initialise-L-P}, we initialize the matrix of solutions $\starpart$ and the set of solved vertices $\solved$:\begin{gather*}
    P = \begin{array}{|c|ccccc|}
    \toprule
    &    i& a& b& e& d\\
    \midrule
    F_1& *& *& *& *& *\\
    \bottomrule
    \end{array}
    \qquad\qquad
    \solved = \emptyset
\end{gather*}

Recall the coefficient block of the system (i.e.\ just the first column) is in row echelon form as a result of Step~\ref{step:gaussian-eliminate-LS}. We now go into the \texttt{while} loop of Step~\ref{step:while-loop}.

In the first loop run, we can identify $e$ and $d$ as the vertices whose corresponding linear systems can be solved (Step~\ref{step:find-L}). This is because the coefficient block is $0$ from the second row onwards and $e, d$ are the only vertices whose columns are also $0$ from the second row onwards. It means that $\tosolve = \{ e, d \}$. We find the solutions of the associated linear systems (Step~\ref{step:solve-L}): for $c$ we must set the variable corresponding to the $F_1$ column to $0$ and for $d$ it must necessarily be $1$. Thus the updated matrix $\starpart$ looks as follows (Step~\ref{step:record-solution-in-P}):\begin{gather*}
    \begin{array}{|c|ccccc|}
    \toprule
    &    i& a& b& e& d\\
    \midrule
    F_1& *& *& *& 0& 1\\
    \bottomrule
    \end{array}
\end{gather*}

Now, we must transform the system of linear equations to the form it would have if $e$ and $d$ had never been included (Step~\ref{step:for-loop}).

Consider $e$ first. In Step~\ref{step:find-v-dependent-rows}, we find which rows of the linear system in \eqref{eq:LS-first-row-echelon} use the original row $e$ from the $e'$ column -- in this case, it is only the fourth row. Thus the fourth row is also the last row using the original $e$ row. (If other rows were using $e$, we would add the fourth row to all other rows using $e$ to remove their dependency on the original $e$ row in Step~\ref{step:row-additions-r1-rk-1}.) Next, in Step~\ref{step:row-addition-rk}, we add the $e$ row of $\ILS$ to the fourth row of \eqref{eq:LS-first-row-echelon}, obtaining the following system (note, that this addition breaks the row echelon form in the coefficient block):\begin{equation}\label{eq:LS-after-row-addition}
    \begin{array}{|c|c|ccccc|ccccc|}
    \toprule
     & F_1 & i & a & b & e & d & i' & a' & b' & e' & d' \\
    \midrule
    1 & 1 & 0 & 1 & 0 & 0 & 1 & 0 & 1 & 0 & 0 & 0 \\
    2 & 0 & 0 & 0 & 0 & 0 & 0 & 1 & 0 & 0 & 0 & 0 \\
    3 & 0 & 0 & 0 & 0 & 0 & 0 & 0 & 0 & 1 & 0 & 0 \\
    4 & 1 & 0 & 1 & 0 & 0 & 1 & 0 & 1 & 0 & 0 & 0 \\
    5 & 0 & 0 & 0 & 0 & 0 & 0 & 0 & 0 & 0 & 0 & 1 \\
    \bottomrule
    \end{array}
\end{equation}

After that, in Step~\ref{step:simplify-rk}, we bring the coefficient block back into row echelon form: we cancel as much in the coefficient block of the fourth row as possible: here, we add the first row of \eqref{eq:LS-after-row-addition} to the fourth row, obtaining the following linear system:\begin{equation}\label{eq:LS-after-e}
    \begin{array}{|c|c|ccccc|ccccc|}
    \toprule
     & F_1 & i & a & b & e & d & i' & a' & b' & e' & d' \\
    \midrule
    1 & 1 & 0 & 1 & 0 & 0 & 1 & 0 & 1 & 0 & 0 & 0 \\
    2 & 0 & 0 & 0 & 0 & 0 & 0 & 1 & 0 & 0 & 0 & 0 \\
    3 & 0 & 0 & 0 & 0 & 0 & 0 & 0 & 0 & 1 & 0 & 0 \\
    4 & 0 & 0 & 0 & 0 & 0 & 0 & 0 & 0 & 0 & 0 & 0 \\
    5 & 0 & 0 & 0 & 0 & 0 & 0 & 0 & 0 & 0 & 0 & 1 \\
    \bottomrule
    \end{array}
\end{equation}

The coefficient block of the matrix is in row echelon form, so we do not have to perform any other operations. In particular, we do not need to swap any rows in Step~\ref{step:swap-rk}, though it might be necessary in general.

Next, we move to vertex $d$. Here the situation is even simpler. From column $d'$ of \eqref{eq:LS-after-e}, we read that only the fifth row uses the original $d$ row. After adding the $d$ row from $\ILS$ to the fifth row (Step~\ref{step:row-addition-rk}), we get the following system, where the coefficient block immediately is in row echelon form:\begin{equation}\label{eq:LS-after-d}
    \begin{array}{|c|c|ccccc|ccccc|}
    \toprule
     & F_1 & i & a & b & e & d & i' & a' & b' & e' & d' \\
    \midrule
    1 & 1 & 0 & 1 & 0 & 0 & 1 & 0 & 1 & 0 & 0 & 0 \\
    2 & 0 & 0 & 0 & 0 & 0 & 0 & 1 & 0 & 0 & 0 & 0 \\
    3 & 0 & 0 & 0 & 0 & 0 & 0 & 0 & 0 & 1 & 0 & 0 \\
    4 & 0 & 0 & 0 & 0 & 0 & 0 & 0 & 0 & 0 & 0 & 0 \\
    5 & 0 & 0 & 0 & 0 & 0 & 0 & 0 & 0 & 0 & 0 & 0 \\
    \bottomrule
    \end{array}
\end{equation}

While dealing with $e$ and $d$, in Step~\ref{step:update-S},  we also add them to the set of solved vertices $\solved$, i.e.\ $\solved = \{ e, d \}$. This ends the first loop run.

In Step~\ref{step:find-L} of the second loop run, we observe that $i$, $a$, and $b$ can all be solved. For $i$ and $b$ the value corresponding to variable $F_1$ must be $0$, and for $a$ it must necessarily be $1$ (Step~\ref{step:solve-L}). It leads to the following matrix $\starpart$ (Step~\ref{step:record-solution-in-P}):\begin{gather*}
    \begin{array}{|c|ccccc|}
    \toprule
    &    i& a& b& e& d\\
    \midrule
    F_1& 0& 1& 0& 0& 1\\
    \bottomrule
    \end{array}
\end{gather*}

While this is no longer necessary to find a solution, we nevertheless also show how the linear systems associated with these vertices would be removed from consideration here as this is what the algorithm will do. First, we `remove' the row of $i$ in Steps~\ref{step:find-v-dependent-rows}--\ref{step:swap-rk}. From the $i'$ column of \eqref{eq:LS-after-d}, we know that only the second row uses the $i$ row from $\ILS$. Adding the $i$ row from $\ILS$ to the second row of the linear system in \eqref{eq:LS-after-d}, we obtain:\begin{equation}\label{eq:LS-after-i}
    \begin{array}{|c|c|ccccc|ccccc|}
    \toprule
     & F_1 & i & a & b & e & d & i' & a' & b' & e' & d' \\
    \midrule
    1 & 1 & 0 & 1 & 0 & 0 & 1 & 0 & 1 & 0 & 0 & 0 \\
    2 & 0 & 0 & 0 & 0 & 0 & 0 & 0 & 0 & 0 & 0 & 0 \\
    3 & 0 & 0 & 0 & 0 & 0 & 0 & 0 & 0 & 1 & 0 & 0 \\
    4 & 0 & 0 & 0 & 0 & 0 & 0 & 0 & 0 & 0 & 0 & 0 \\
    5 & 0 & 0 & 0 & 0 & 0 & 0 & 0 & 0 & 0 & 0 & 0 \\
    \bottomrule
    \end{array}
\end{equation}

The coefficient block is in row echelon form. We move to vertex $a$. From the $a'$ column, we read that only the first row uses the $a$ row from $\ILS$. We thus add the $a$ row from $\ILS$ to the first row in \eqref{eq:LS-after-i}, obtaining:\begin{equation}\label{eq:LS-after-a}
    \begin{array}{|c|c|ccccc|ccccc|}
    \toprule
     & F_1 & i & a & b & e & d & i' & a' & b' & e' & d' \\
    \midrule
    1 & 0 & 0 & 0 & 0 & 0 & 0 & 0 & 0 & 0 & 0 & 0 \\
    2 & 0 & 0 & 0 & 0 & 0 & 0 & 0 & 0 & 0 & 0 & 0 \\
    3 & 0 & 0 & 0 & 0 & 0 & 0 & 0 & 0 & 1 & 0 & 0 \\
    4 & 0 & 0 & 0 & 0 & 0 & 0 & 0 & 0 & 0 & 0 & 0 \\
    5 & 0 & 0 & 0 & 0 & 0 & 0 & 0 & 0 & 0 & 0 & 0 \\
    \bottomrule
    \end{array}
\end{equation}

The coefficient block is in row echelon form, so we move to $b$. Again, we find that only one row of \eqref{eq:LS-after-a} -- the third row -- uses the original $b$ row. After adding the $b$ row from $\ILS$ to the third row of the linear system in \eqref{eq:LS-after-a}, we get:\begin{gather*}
    \begin{array}{|c|c|ccccc|ccccc|}
    \toprule
     & F_1 & i & a & b & e & d & i' & a' & b' & e' & d' \\
    \midrule
    1 & 0 & 0 & 0 & 0 & 0 & 0 & 0 & 0 & 0 & 0 & 0 \\
    2 & 0 & 0 & 0 & 0 & 0 & 0 & 0 & 0 & 0 & 0 & 0 \\
    3 & 0 & 0 & 0 & 0 & 0 & 0 & 0 & 0 & 0 & 0 & 0 \\
    4 & 0 & 0 & 0 & 0 & 0 & 0 & 0 & 0 & 0 & 0 & 0 \\
    5 & 0 & 0 & 0 & 0 & 0 & 0 & 0 & 0 & 0 & 0 & 0 \\
    \bottomrule
    \end{array}
\end{gather*}

In the process, at Step~\ref{step:update-S}, vertices $i,a,b$ are added to set $\solved$. Hence, $\solved = \{ i,a,b,e,d \}$ and the loop ends here.

Finally, in Step~\ref{step:build-c-B}, we construct the matrix $\colbasechange{C}{\mathcal{B}} = \left[ \frac{Id_{\comp{O}}}{\starpart} \right]$, i.e.:\begin{gather*}
    \colbasechange{C}{\mathcal{B}} = \begin{array}{|c|ccccc|}
        \toprule
         &   i& a& b& e& d\\
        \midrule
        i&   1& 0& 0& 0& 0\\
        a&   0& 1& 0& 0& 0\\
        b&   0& 0& 1& 0& 0\\
        e&   0& 0& 0& 1& 0\\
        d&   0& 0& 0& 0& 1\\
        \midrule
        F_1& 0& 1& 0& 0& 1\\
        \bottomrule
    \end{array}
\end{gather*}

We finish the computation with Step~\ref{step:output} by computing the correction matrix $C=C'\colbasechange{C}{\mathcal{B}}$ encoding the correction function, and the product $NC = NC'\colbasechange{C}{\mathcal{B}}$ encoding the induced relation.
The two matrices are the same as those shown in Figure~\ref{fig:mats}.

\end{document}